\documentclass[11pt,a4paper]{article}

\usepackage{fullpage}
\usepackage[utf8x]{inputenc}
\usepackage{amsmath, amsthm}
\usepackage{wrapfig,graphicx,amssymb,textcomp,array,amsmath}
\usepackage{algpseudocode} 
\usepackage{enumitem}
\algtext*{EndWhile}
\algtext*{EndIf}
\usepackage{algorithm}

\newcommand{\ConvexShape}{\bigtriangledown}
\newcommand{\kTD}[2]{$#1$\text{-}TD#2}
\newcommand{\kDT}[2]{$#1$\text{-}DG#2}
\newcommand{\kGG}[2]{$#1$\text{-}GG#2}
\newcommand{\kRNG}[2]{$#1$\text{-}RNG#2}
\newcommand{\Tm}{t_{min}}

\newcommand{\WS}[1]{\text{WS$(#1)$}}

\title{Higher-Order Triangular-Distance Delaunay Graphs: Graph-Theoretical Properties}

\author{
Ahmad Biniaz\thanks{School of Computer Science, Carleton University, 
                    Ottawa, Canada. Research supported by NSERC.}
\and 
Anil Maheshwari\footnotemark[1]
\and 
Michiel Smid\footnotemark[1]
}
\date{\today}
\newtheorem{lemma}{Lemma}
\newtheorem{corollary}{Corollary}
\newtheorem{theorem}{Theorem}
\newtheorem{observation}{Observation}
\begin{document}

\maketitle

\begin{abstract}
We consider an extension of the triangular-distance Delaunay graphs (TD-Delaunay) on a set $P$ of points in the plane. In TD-Delaunay, the convex distance is defined by a fixed-oriented equilateral triangle $\ConvexShape$, and there is an edge between two points in $P$ if and only if there is an empty homothet of $\ConvexShape$ having the two points on its boundary. We consider higher-order triangular-distance Delaunay graphs, namely \kTD{k}{}, which contains an edge between two points if the interior of the homothet of $\ConvexShape$ having the two points on its boundary contains at most $k$ points of $P$. We consider the connectivity, Hamiltonicity and perfect-matching admissibility of \kTD{k}{}. Finally we consider the problem of blocking the edges of \kTD{k}{}.
\end{abstract}

\section{Introduction}
The {\em triangular-distance Delaunay graph} of a point set $P$ in the plane, TD-Delaunay for short, was introduced by Chew \cite{Chew1989}. A TD-Delaunay is a graph whose convex distance function is defined by a fixed-oriented equilateral triangle. Let $\ConvexShape$ be a downward equilateral triangle whose barycenter
is the origin and one of its vertices is on negative $y$-axis. A {\em homothet} of $\ConvexShape$ is obtained by scaling $\ConvexShape$ with respect to the origin by some factor $\mu\ge 0$, followed by a translation to a point $b$ in the plane: $b+\mu\ConvexShape=\{b+\mu a:a\in\ConvexShape\}$.
In the TD-Delaunay graph of $P$, there is a straight-line edge between two points $p$ and $q$ if and only if there exists a homothet of $\ConvexShape$ having $p$ and $q$ on its boundary and whose interior does not contain any point of $P$. In other words, $(p,q)$ is an edge of TD-Delaunay graph if and only if there exists an empty downward equilateral triangle having $p$ and $q$ on its boundary. In this case, we say that the edge $(p,q)$ has the {\em empty triangle property}. 
The TD-Delaunay graph is a planar graph, see \cite{Bose2010}.
We define $t(p,q)$ as the smallest homothet of $\ConvexShape$ having $p$ and $q$ on its boundary. See Figure~\ref{TD}(a). Note that $t(p,q)$ has one of $p$ and $q$ at a vertex, and the other one on the opposite side. Thus,

\begin{observation}
\label{side-point-obs}
 Each side of $t(p,q)$ contains either $p$ or $q$.
\end{observation}

In \cite{Babu2013}, the authors proved a tight lower bound of $\lceil\frac{n-1}{3}\rceil$ on the size of a maximum matching in a TD-Delaunay graph. In this paper we study higher-order TD-Delaunay graphs. An {\em order-k TD-Delaunay graph} of a point set $P$, denoted by \kTD{k}{}, is a geometric graph which has an edge $(p,q)$ iff the interior of $t(p,q)$ contains at most $k$ points of $P$; see Figure~\ref{TD}(b). The standard TD-Delaunay graph corresponds to \kTD{0}{}. We consider graph-theoretic properties of higher-order TD-Delaunay graphs, such as connectivity, Hamiltonicity, and perfect-matching admissibility. We also consider the problem of blocking TD-Delaunay graphs.

\begin{figure}[htb]
  \centering
\setlength{\tabcolsep}{0in}
  $\begin{tabular}{ccc}
\multicolumn{1}{m{.33\columnwidth}}{\centering\includegraphics[width=.3\columnwidth]{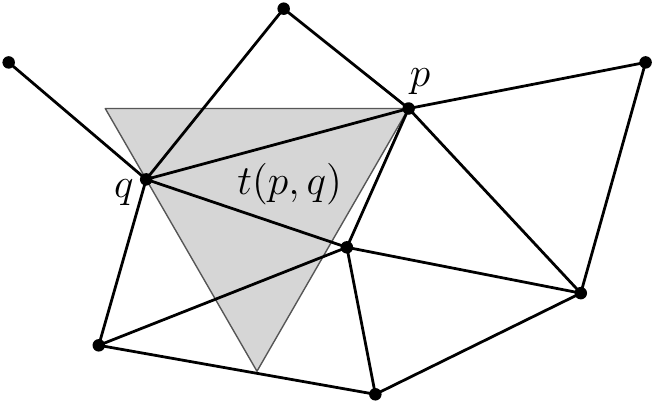}}
&\multicolumn{1}{m{.33\columnwidth}}{\centering\includegraphics[width=.3\columnwidth]{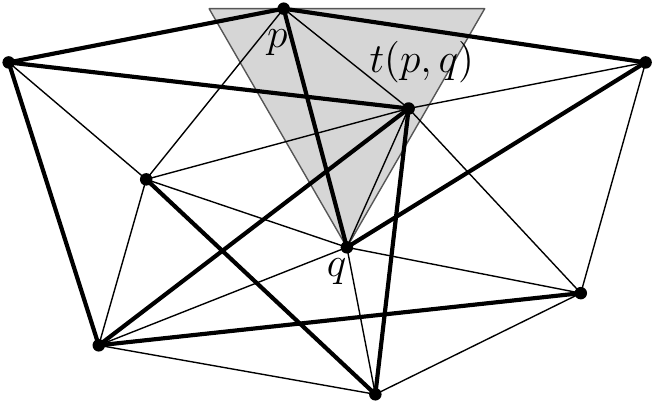}} &\multicolumn{1}{m{.33\columnwidth}}{\centering\includegraphics[width=.3\columnwidth]{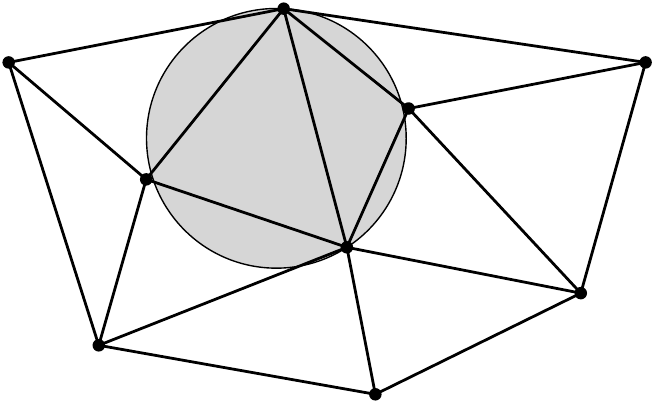}}
\\
(a) & (b)& (c)
\end{tabular}$
  \caption{(a) Triangular-distance Delaunay graph (\kTD{0}{}), (b) \kTD{1}{} graph, the light edges belong to \kTD{0}{} as well, and (c) Delaunay triangulation.}
\label{TD}
\end{figure}

\subsection{Previous Work}
\label{previous-work}
A {\em Delaunay triangulation} (DT) of $P$ is a graph whose distance function is defined by a fixed circle {\footnotesize $\bigcirc$} centered at the origin. DT has an edge between two points $p$ and $q$ if there exists a homothet of {\footnotesize $\bigcirc$} having $p$ and $q$ on its boundary and whose interior does not contain any point of $P$; see Figure~\ref{TD}(c). In this case the edge $(p,q)$ is said to have the {\em empty circle property}. An {\em order-k
Delaunay Graph} on $P$, denoted by \kDT{k}{}, is defined to have an edge $(p, q)$ iff there exists a homothet of {\footnotesize $\bigcirc$} having $p$ and $q$ on its boundary and whose interior contains at most $k$ points of $P$. The standard Delaunay triangulation corresponds to \kDT{0}{}.  

For each pair of points $p,q\in P$ let $D[p,q]$ be the closed disk having $pq$ as diameter. A {\em Gabriel Graph} on $P$ is a geometric graph which has an edge between two points $p$ and $q$ iff $D[p,q]$ does not contain any point of $P\setminus\{p,q\}$. An {\em order-$k$ Gabriel Graph} on $P$, denoted by \kGG{k}{}, is defined to have an edge $(p,q)$ iff $D[p,q]$ contains at most $k$ points of $P\setminus\{p,q\}$.

For each pair of points $p,q\in P$, let $L(p,q)$ be the intersection of the two open disks with radius $|pq|$ centered at $p$ and $q$. A {\em Relative Neighborhood Graph} on $P$ is a geometric graph which has an edge between two points $p$ and $q$ iff $L(p,q)$ does not contain any point of $P$. An {\em order-$k$ Relative Neighborhood Graph} on $P$, denoted by \kRNG{k}{}, is defined to have an edge $(p,q)$ iff $L(p,q)$ contains at most $k$ points of $P$. It is obvious that $\text{\kRNG{k}{}}\subseteq\text{\kGG{k}{}}\subseteq\text{\kDT{k}{}}$.

The problem of determining whether an order-$k$ geometric graph always has a (bottleneck) perfect matching or a (bottleneck) Hamiltonian cycle is quite of interest. We will define these notions in Section~\ref{graph-notions}. 
Chang et al. \cite{Chang1992b, Chang1992, Chang1991} proved that a Euclidean bottleneck biconnected spanning graph, bottleneck perfect matching, and bottleneck Hamiltonian cycle of $P$ are contained in \kRNG{1}{}, \kRNG{16}{}, \kRNG{19}{}, respectively. This implies that \kRNG{16}{} has a perfect matching and \kRNG{19}{} is Hamiltonian. Since \kRNG{k}{} is a subgraph of \kGG{k}{}, the same results hold for \kGG{16}{} and \kGG{19}{}. It is known that \kGG{k}{} is $(k+1)$-connected \cite{Bose2013} and 
\kGG{15}{} (and hence \kDT{15}{}) is Hamiltonian. Dillencourt showed that a Delaunay triangulation (\kDT{0}{}) admits a perfect matching \cite{Dillencourt1990} but it can fail to be Hamiltonian \cite{Dillencourt1987a}. 

Given a geometric graph $G(P)$ on a set $P$ of $n$ points, we say that a set $K$ of points blocks $G(P)$ if in $G(P\cup K)$ there is no edge connecting two points in $P$. Actually $P$ is an independent set in $G(P\cup K)$.
Aichholzer et al.~\cite{Aichholzer2013} considered the problem of blocking the Delaunay triangulation (i.e. \kDT{0}{}) for $P$ in general position. They show that $\frac{3n}{2}$ points are sufficient to block DT($P$) and at least $n-1$ points are necessary. To block a Gabriel graph, $n-1$ points are sufficient \cite{Aronov2013}.

In a companion paper, we considered the matching and blocking problems in higher-order Gabriel graphs. We showed that \kGG{10}{} contains a Euclidean bottleneck matching and \kGG{8}{} may not have any. As for maximum matching, we proved a tight lower bound of $\frac{n-1}{4}$ in \kGG{0}{}. We also showed that \kGG{1}{} has a matching of size at least $\frac{2(n-1)}{5}$ and \kGG{2}{} has a perfect matching (when $n$ is even). In addition, we showed that $\lceil\frac{n-1}{3}\rceil$ points are necessary to block \kTD{0}{} and this bound is tight.
\subsection{Our Results}
\label{our-results}
We show for which values of $k$, \kTD{k}{} contains a bottleneck biconnected spanning graph, bottleneck Hamiltonian cycle, and (bottleneck) perfect-matching. We define these notions Section~\ref{graph-notions}. In Section~\ref{connectivity} we prove that every \kTD{k}{} graph is $(k+1)$-connected. In addition we show that a bottleneck biconnected spanning graph of $P$ is contained in \kTD{1}{}. Using a similar approach as in \cite{Abellanas2009, Chang1991}, in Section~\ref{Hamiltonicity} we show that a bottleneck Hamiltonian cycle of $P$ is contained in \kTD{7}{}. We also show a configuration of a point set $P$ such that \kTD{5}{} fails to have a bottleneck Hamiltonian cycle. In Section~\ref{matching} we prove that a bottleneck perfect matching of $P$ is contained in \kTD{6}{}, and we show that for some point set $P$, \kTD{5}{} does not have a bottleneck perfect matching. In Section~\ref{matching2} we prove that \kTD{2}{} has a perfect matching and \kTD{1}{} has a matching of size at least $\frac{2(n-1)}{5}$. In Section~\ref{blocking-section} we consider the problem of blocking \kTD{k}{}. We show that at least $\lceil\frac{n-1}{2}\rceil$ points are necessary and $n-1$ points are sufficient to block a \kTD{0}{}. The open problems and concluding remarks are presented in Section~\ref{conclusion}.

\section{Preliminaries}
\label{preliminaries}
\subsection{Some Geometric Notions}
\label{geometric-notions}
Bonichon et al. \cite{Bonichon2010} showed that a half-$\Theta_6$ graph of a point set $P$ in the plane is equal to a TD-Delaunay graph of $P$. They also showed that every plane triangulation is TD-Delaunay realizable. 

\begin{wrapfigure}{R}{0.35\textwidth}
  \begin{center}
\includegraphics[width=.3\textwidth]{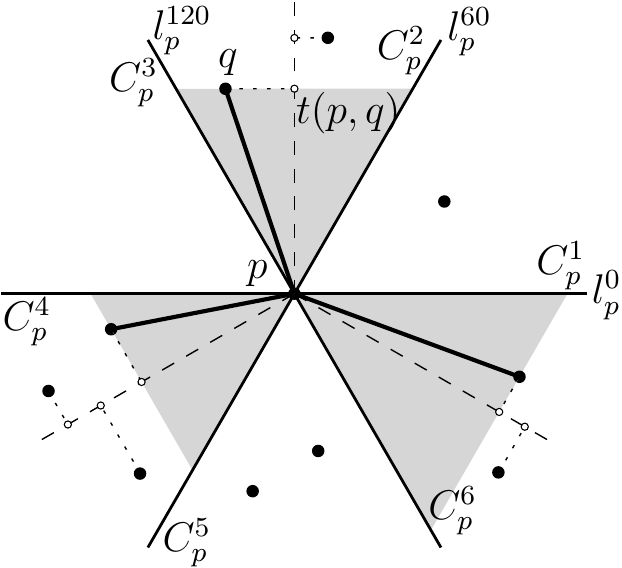}
  \end{center}
  \caption{The construction of the TD-Delaunay graph.}
\label{cones}
\end{wrapfigure}

A half-$\Theta_6$ graph (or equivalently a TD-Delaunay graph) on a point set $P$ can be constructed in the following way. For each point $p$ in $P$, let $l_p$ be the horizontal line through $p$. Define $l_p^{\gamma}$ as the line obtained by rotating $l_p$ by $\gamma$-degrees in counter-clockwise direction around $p$. Actually $l_p^0=l_p$. Consider three lines $l_p^{0}$, $l_p^{60}$, and $l_p^{120}$ which partition the plane into six disjoint cones with apex $p$. Let $C_p^1, \dots, C_p^6$ be the cones in counter-clockwise order around $p$ as shown in Figure~\ref{cones}. We partition the cones into the set of {\em odd cones} $\{C_p^1,C_p^3,C_p^5\}$, and the set of {\em even cones} $\{C_p^2,C_p^4,C_p^6\}$. For each even cone $C_p^i$ connect $p$ to the ``nearest'' point $q$ in $C_p^i$. The {\em distance} between $p$ and $q$, $d(p,q)$, is defined as the Euclidean distance between $p$ and the orthogonal projection of $q$ onto the bisector of $C_p^i$. See Figure~\ref{cones}. The resulting graph is the half-$\Theta_6$ graph which is defined by even cones \cite{Bonichon2010}. Moreover, the resulting graph is the TD-Delaunay graph defined with respect to homothets of $\ConvexShape$. By considering the odd cones, another half-$\Theta_6$ graph is obtained. The well-known $\Theta_6$ graph is the union of half-$\Theta_6$ graphs defined by odd and even cones. To construct \kTD{k}{}, for each point $p\in P$ we connect $p$ to its $(k+1)$ nearest neighbors in each even cone around $p$. It is obvious that \kTD{k}{} has $O(kn)$ edges. The \kTD{k}{} can be constructed in $O(n\log n+kn\log \log n)$-time, using the algorithm introduced by Lukovszki~\cite{Lukovszki1999} for computing fault tolerant spanners.

Recall that $t(p,q)$ is the smallest homothet of $\ConvexShape$ having $p$ and $q$ on its boundary. In other words, $t(p,q)$ is the smallest downward equilateral triangle through $p$ and $q$. Similarly we define $t'(p,q)$ as the smallest upward equilateral triangle having $p$ and $q$ on its boundary. It is obvious that the even cones correspond to downward triangles and odd cones correspond to upward triangles.   
We define an order on the equilateral triangles: for each two equilateral triangles $t_1$ and $t_2$ we say that $t_1<t_2$ if the area of $t_1$ is less than the area of $t_2$. Since the area of $t(p,q)$ is directly related to $d(p,q)$, 
$$d(p,q)<d(r,s) \quad\text{ if and only if }\quad t(p,q)<t(r,s).$$

\begin{figure}[htb]
  \centering
\setlength{\tabcolsep}{0in}
  $\begin{tabular}{ccc}
\multicolumn{1}{m{.33\columnwidth}}{\centering\includegraphics[width=.2\columnwidth]{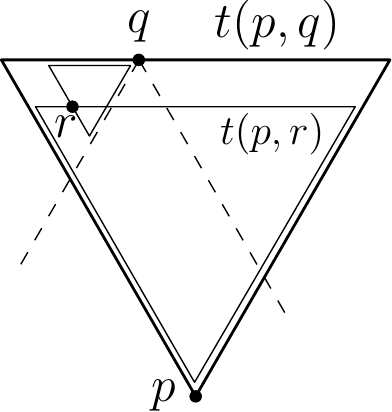}}
&\multicolumn{1}{m{.33\columnwidth}}{\centering\includegraphics[width=.2\columnwidth]{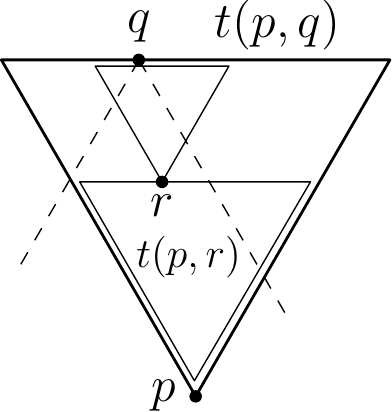}} &\multicolumn{1}{m{.33\columnwidth}}{\centering\includegraphics[width=.2\columnwidth]{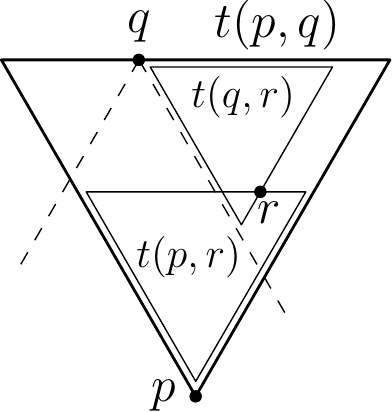}}
\\
(a) & (b)& (c)
\end{tabular}$
  \caption{Illustration of Observation~\ref{obs1}: the point $r$ is contained in $t(p,q)$. The triangles $t(p,r)$ and $t(q,r)$ are inside $t(p,q)$.}
\label{smaller-triangle-fig}
\end{figure}

As shown in Figure~\ref{smaller-triangle-fig} we have the following observation:

\begin{observation}
\label{obs1}
 If $t(p,q)$ contains a point $r$, then $t(p,r)$ and $t(q,r)$ are contained in $t(p,q)$.
\end{observation}
As a direct consequence of Observation \ref{obs1}, if a point $r$ is contained in $t(p,q)$, then $\max\{t(p,r), \allowbreak t(q,r)\}<t(p,q)$. It is obvious that,

\begin{observation}
 \label{equal-triangles}
 For each two points $p,q\in P$, $t(p,q)=t'(p,q)$.
\end{observation}
Thus, we define $X(p,q)$ as a regular hexagon centred at $p$ which has $q$ on its boundary, and its sides are parallel to $l_p^0$, $l_p^{60}$, and $l_p^{120}$. 
\begin{observation}
\label{obs2}
 If $X(p,q)$ contains a point $r$, then $t(p,r)<t(p,q)$.
\end{observation}
For each edge $(p,q)$ in \kTD{k}{} we define its {\em weight}, $w(p,q)$, to be equal to the area of $t(p,q)$.
\subsection{Some Graph-Theoretic Notions}
\label{graph-notions}
A graph $G$ is {\em connected} if there is a path between any pair of vertices in $G$. Moreover, $G$ is $k$-{\em connected} if there does not exist a set of at most $k-1$ vertices whose removal disconnects $G$. In case $k=2$, $G$ is called {\em biconnected}. In other words a graph $G$ is biconnected iff there is a simple cycle between any pair of its vertices. A {\em matching} in $G$ is a set of edges in $G$ without common vertices. A {\em perfect matching} is a matching which matches all the vertices of $G$. A {\em Hamiltonian cycle} in $G$ is a cycle (i.e., closed loop) through $G$ that visits each vertex of $G$ exactly once.
In case that $G$ is an edge-weighted graph, a {\em bottleneck matching} (resp. {\em bottleneck Hamiltonian cycle}) is defined to be a perfect matching (resp. Hamiltonian cycle) in $G$ with the weight of the maximum-weight edge is minimized. A {\em bottleneck biconnected spanning subgraph} of $G$ is a spanning subgraph, $G'$, of $G$ which is biconnected and the weight of the longest edge in $G'$ is minimized. For $H\subseteq G$ we denote the bottleneck of $H$, i.e., the length of the longest edge in $H$, by $\lambda(H)$.

For a graph $G=(V,E)$ and $K\subseteq V$, let $G-K$ be the subgraph obtained from $G$ by removing vertices in $K$, and let $o(G-K)$ be the number of odd components in $G-K$. The following theorem by Tutte~\cite{Tutte1947} gives a characterization of the graphs which have perfect matching: 

\begin{theorem}[Tutte~\cite{Tutte1947}] 
\label{Tutte} 
$G$ has a perfect matching if and only if $o(G-K)\le |K|$ for all $K\subseteq V$.
\end{theorem}

Berge~\cite{Berge1958} extended Tutte’s theorem to a formula (known as Tutte-Berge formula) for the maximum size of a matching in a graph. In a graph $G$, the {\em deficiency}, $\text{def}_G(K)$, is $o(G-K)-|K|$. Let $\text{def}(G)=\max_{K\subseteq V}{\text{def}_G(K)}$.

\begin{theorem}[Tutte-Berge formula; Berge~\cite{Berge1958}] 
\label{Berge} 
The size of a maximum matching in $G$ is $$\frac{1}{2}(n-\mathrm{def}(G)).$$
\end{theorem}

For an edge-weighted graph $G$ we define the {\em weight sequence} of $G$, \WS{G}, as the sequence containing the weights of the edges of $G$ in non-increasing order. A graph $G_1$ is said to be less than a graph $G_2$ if \WS{G_1} is lexicographically smaller than \WS{G_2}.

\section{Connectivity}
\label{connectivity}
In this section we consider the connectivity of higher-order triangular-distance Delaunay graphs.

\subsection{$(k+1)$-connectivity}
\label{connectivity-k-plus-1}
For a set $P$ of points in the plane, the TD-Delaunay graph, i.e., \kTD{0}{}, is not necessarily a triangulation \cite{Chew1989}, but it is connected and internally triangulated \cite{Babu2013}. As shown in Figure~\ref{TD}(a), the outer face may not be convex and hence \kTD{0}{} is not necessarily biconnected. As a warm up exercise we show that every \kTD{k}{} is $(k+1)$-connected. 

\begin{theorem}
\label{k-connectivity-thr}
 For every point set $P$, \kTD{k}{} is $(k+1)$-connected. In addition, for every $k$, there exists a point set $P$ such that \kTD{k}{} is not $(k+2)$-connected.
\end{theorem}
\begin{proof}
We prove the first part of this theorem by contradiction. Let $K$ be the set of (at most) $k$ vertices removed from \kTD{k}{}, and let $\mathcal{C}=\{C_1, C_2, \dots, C_m\}$, where $m>1$, be the resulting maximal connected components. Let $T$ be the set of all triangles defined by any pair of points belonging to different components, i.e., $T=\{t(a,b): a\in C_i, b\in C_j, i\neq j\}$. Consider the smallest triangle $\Tm\in T$. Assume that $\Tm$ is defined by two points $a$ and $b$, i.e., $\Tm=t(a,b)$, where $a\in C_i$, $b\in C_j$, and $i\neq j$.

{\em Claim 1}: $\Tm$ does not contain any point of $P\setminus K$ in its interior.
By contradiction, suppose that $\Tm$ contains a point $c\in P\setminus K$ in its interior. Three cases arise: (i) $c\in C_i$, (ii) $c\in C_j$, (iii) $c\in C_l$, where $l\neq i$ and $l\neq j$. In case (i) the triangle $t(c,b)$ between $C_i$ and $C_j$ is contained in $t(a,b)$. In case (ii) the triangle $t(a,c)$ between $C_i$ and $C_j$ is contained in $t(a,b)$. In case (iii) both triangles $t(a,c)$ and  $t(c,b)$ are contained in $t(a,b)$. All cases contradict the minimality of $t(a,b)=\Tm$. Thus, $\Tm$ contains no point of $P\setminus K$ in its interior, proving Claim 1.

By Claim 1, $\Tm=t(a,b)$ may only contain points of $K$. Since $|K|\le k$, there must be an edge between $a$ and $b$ in \kTD{k}{}. This contradicts that $a$ and $b$ belong to different components $C_i$ and $C_j$ in $\mathcal{C}$. Therefore, \kTD{k}{} is $(k+1)$-connected.

We present a constructive proof for the second part of theorem. Let $P=A\cup B\cup K$, where $|A|,|B|\ge 1$ and $|K|=k+1$. Place the points of $A$ in the plane. Let $C_A^{4}=\bigcap_{p\in A}{C_p^4}$. Place the points of $K$ in $C_A^4$. Let $C_K^4=\bigcap_{p\in K}{C_p^4}$. Place the points of $B$ in $C_K^4$. Consider any pair $(a,b)$ of points where $a\in A$ and $b\in B$. It is obvious that any path between $a$ and $b$ in \kTD{k}{} goes through the vertices in $K$. Thus by removing the vertices in $K$, $a$ and $b$ become disconnected. Therefore, \kTD{k}{} of $P$ is not $(k+2)$-connected. 
\end{proof}
\subsection{Bottleneck Biconnected Spanning Graph}
As shown in Figure~\ref{TD}(a), \kTD{0}{} may not be biconnected. By Theorem~\ref{k-connectivity-thr}, \kTD{1}{} is biconnected. In this section we show that a bottleneck biconnected spanning graph of $P$ is contained in \kTD{1}{}. 

\begin{theorem}
\label{biconnectivity-thr}
 For every point set $P$, \kTD{1}{} contains a bottleneck biconnected spanning graph of $P$.
\end{theorem}

\begin{proof}                                                                                       
Let $\mathcal{G}$ be the set of all biconnected spanning graphs with vertex set $P$. We define a total order on the elements of $\mathcal{G}$ by their weight sequence. If two elements have the same weight sequence, we break the ties arbitrarily to get a total order.
Let $G^* = (P, E)$ be a graph in $\mathcal{G}$ with minimal weight sequence. Clearly, $G^*$ is a bottleneck biconnected spanning graph of $P$. We will show that all edges of $G^*$ are in \kTD{1}{}. By contradiction suppose that some edges in $E$ do not belong to \kTD{1}{}, and let $e = (a, b)$ be the longest one (by the area of the triangle $t(a,b)$). If the graph $G^*-\{e\}$ is biconnected, then by removing $e$, we obtain a biconnected spanning graph $G$ with $\WS{G}<\WS{G^*}$; contradicting the minimality of $G^*$. Thus, there is a pair $(p,q)$ of points such that any cycle between $p$ and $q$ in $G^*$ goes through $e$. Since $(a,b)\notin$~\kTD{1}{}, $t(a,b)$ contains at least two points of $P$, say $x$ and $y$. Let $G$ be the graph obtained from $G^*$ by removing the edge $(a,b)$ and adding the edges $(a,x)$, $(b,x)$, $(a,y)$, $(b,y)$. We show that in $G$ there is a cycle between $p$ and $q$ which does not go through $e$. Consider a cycle $C$ in $G^*$ between two points $p$ and $q$ (which goes through $e$). If none of $x$ and $y$ belong to $C$, then $(C-\{(a,b)\})\allowbreak\cup\allowbreak\{(a,x),(b,x)\}$ is a cycle in $G$ between $p$ and $q$. If one of $x$ or $y$, say $x$, belongs to $C$, then $(C-\{(a,b)\})\allowbreak\cup\allowbreak\{(a,y),(b,y)\}$ is a cycle in $G$ between $p$ and $q$. If both $x$ and $y$ belong to $C$, consider the partition of $C$ into four parts: (a) edge $(a,b)$, (b) path $\delta_{bx}$ between $b$ and $x$, (c) path $\delta_{xy}$ between $x$ and $y$, and (d) path $\delta_{ya}$ between $y$ and $a$. 
There are four cases:
\begin{enumerate}
 \item None of $p$ and $q$ are on $\delta_{xy}$. Then $\delta_{bx}\cup \delta_{ya} \cup \{(a,x),(b,y)\}$ is a cycle in $G$ between $p$ and $q$.
 \item Both $p$ and $q$ are on $\delta_{xy}$. Then $\delta_{xy}\cup \{(a,x),(a,y)\}$ is a cycle in $G$ between $p$ and $q$.
 \item One of $p$ and $q$ is on $\delta_{xy}$ and the other one is on $\delta_{bx}$. Then $\delta_{bx}\cup\delta_{xy}\cup\{(b,y)\}$ is a cycle in $G$ between $p$ and $q$.
 \item One of $p$ and $q$ is on $\delta_{xy}$ and the other one is on $\delta_{ya}$. Then $\delta_{xy}\cup\delta_{ya}\cup\{(a,x)\}$ is a cycle in $G$ between $p$ and $q$.
\end{enumerate}

Thus, between any pair of points in $G$ there exists a cycle, and hence $G$ is biconnected. Since $x$ and $y$ are inside $t(a,b)$, by Observation~\ref{obs1}, $\max\{t(a,x), t(a, y), t(b,x),t(b,y)\}<t(a,b)$. Therefore, $\WS{G}<\WS{G^*}$; contradicting the minimality of $G^*$.   
\end{proof}

\section{Hamiltonicity}
\label{Hamiltonicity}

In this section we show that \kTD{7}{} contains a bottleneck Hamiltonian cycle. In addition, we will show that for some point sets, \kTD{5}{} does not contain any bottleneck Hamiltonian cycle.

\begin{theorem}
\label{hamiltonicity-thr}
 For every point set $P$, \kTD{7}{} contains a bottleneck Hamiltonian cycle.
\end{theorem}

\begin{proof}
Let $\mathcal{H}$ be the set of all Hamiltonian cycles through the points of $P$. Define a total
order on the elements of $\mathcal{H}$ by their weight sequence. If two elements have exactly the same weight sequence, break ties arbitrarily to get a total order. 
Let $H^* = a_0, a_1,\dots, a_{n−1}$ be a cycle in $\mathcal{H}$ with minimal weight sequence. It is obvious that $H^*$ is a bottleneck Hamiltonian cycle of $P$. We will show that all the edges of $H^*$ are in \kTD{7}{}. Consider any edge $e = (a_i, a_{i+1})$ in $H^*$ and let $t(a_i,a_{i+1})$ be the triangle corresponding to $e$ (all index manipulations are modulo $n$).

{\em Claim 1}: None of the edges of $H^*$ can be completely inside $t(a_i,a_{i+1})$. Suppose there is an edge $f=(a_j, a_{j+1})$ inside $t(a_i,a_{i+1})$. Let $H$ be a cycle obtained from $H^*$ by deleting $e$ and $f$, and adding $(a_i, a_j)$ and $(a_{i+1}, a_{j+1})$. By Observation~\ref{obs1}, $t(a_i, a_{i+1}) > \max\{t(a_i, a_j), t(a_{i+1}, a_{j+1})\}$, and hence $\WS{H}<\WS{H^*}$. This contradicts the minimality of $H^*$.

Therefore, we may assume that no edge of $H^*$ lies completely inside $t(a_i,a_{i+1})$. Suppose there are $w$ points of $P$ inside $t(a_i,a_{i+1})$. Let $U = u_1, u_2,\dots, u_w$ represent these points indexed in the order we would encounter them on $H^*$ starting from $a_i$. Let $S = s_1, s_2,\dots, s_w$ and $R = r_1, r_2,\dots, r_w$ represent the vertices where $s_i$ is the vertex preceding $u_i$ on the cycle and $r_i$ is the vertex succeeding $u_i$ on the cycle.
Without loss of generality assume that $a_i\in C_{a_{i+1}}^4$, and $t(a_i,a_{i+1})$ is anchored at $a_{i+1}$, as shown in Figure~\ref{hamiltonicity-fig1}. 

{\em Claim 2}: For each $r_j\in R$, $t(r_j, a_{i+1}) \ge \max\{t(a_i, a_{i+1}), t(u_j, r_j)\}$. Suppose there is a point $r_j\in R$ such that $t(r_j, a_{i+1}) < \max\{t(a_i, a_{i+1}), t(u_j, r_j)\}$. Construct a new cycle $H$ by removing the edges $(u_j, r_j)$, $(a_i, a_{i+1})$ and adding the edges $(a_{i+1}, r_j)$ and $(a_i, u_j)$. Since the two new edges have length strictly less than $\max\{t(a_i, a_{i+1}), t(u_j, r_j)\}$, $\WS{H} < \WS{H^*}$; which is a contradiction.

{\em Claim 3}: For each pair $r_j$ and $r_k$ of points in $R$, $t(r_j, r_k)\ge \max\{t(a_i, a_{i+1}), t(u_j,r_j),\allowbreak t(u_k,r_k)\}$. Suppose there is a pair $r_j$ and $r_k$ such that $t(r_j, r_k)< \max\{t(a_i, a_{i+1}), t(u_j,r_j),\allowbreak d(u_k,r_k)\}$.  Construct a new cycle $H$ from $H^*$ by first deleting $(u_j,r_j)$, $(u_k,r_k)$,  $(a_i, a_{i+1})$. This results in three paths. One of the paths must contain both $a_i$ and either $r_j$ or $r_k$. W.l.o.g. suppose that $a_i$ and $t_k$ are on the same path. Add the edges $(a_i, u_j)$, $(a_{i+1}, u_k)$, $(r_j, r_k)$. Since $\max\{t(u_j,r_j),\allowbreak t(u_k,r_k),\allowbreak d(a_i, a_{i+1})\}>\max \{t(a_i, u_j),\allowbreak t(a_{i+1}, u_k),\allowbreak t(r_j, r_k)\}$, $\WS{H} <\allowbreak \WS{H^*}$; which is a contradiction.

\begin{figure}[htb]
  \centering
  \includegraphics[width=.6\columnwidth]{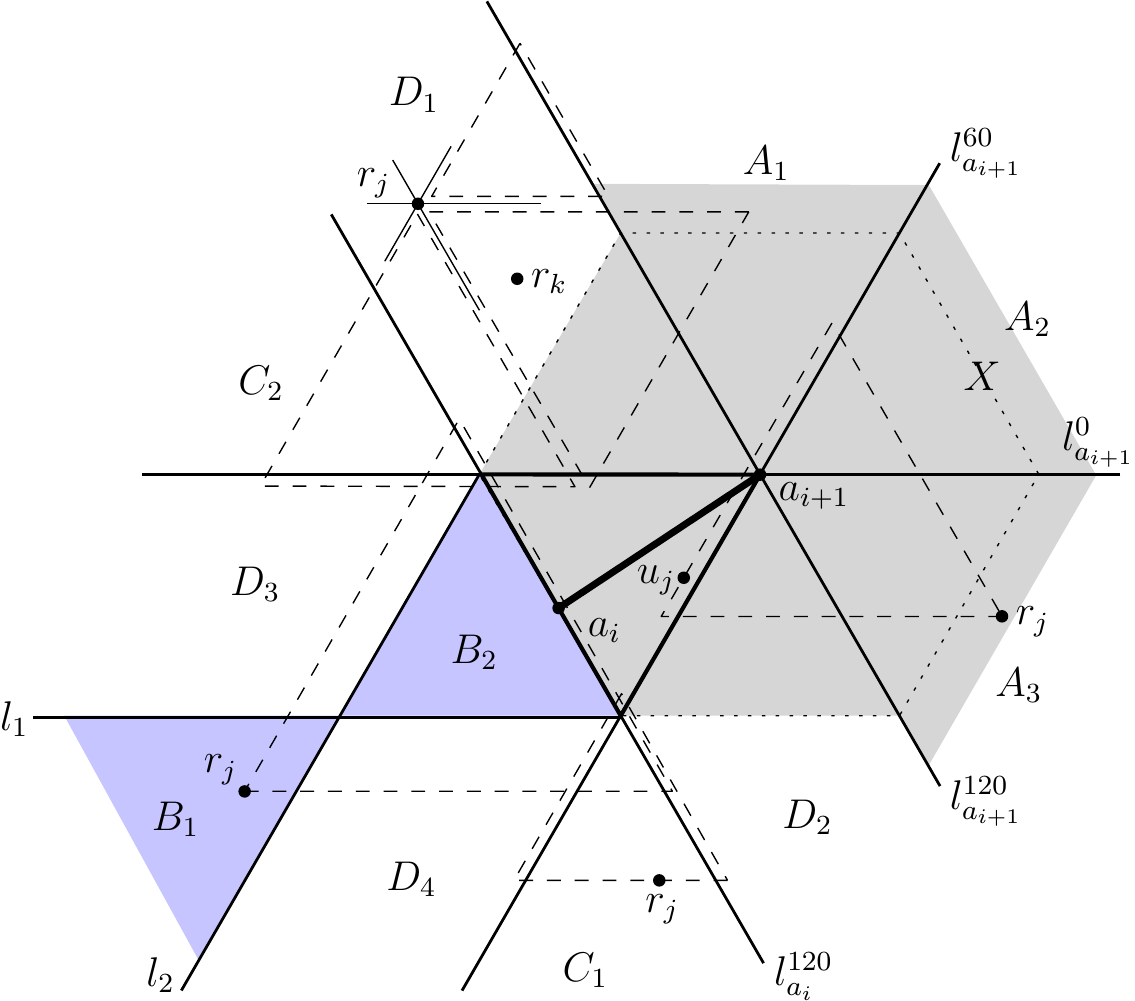}
 \caption{Illustration of Theorem~\ref{hamiltonicity-thr}.}
  \label{hamiltonicity-fig1}
\end{figure}

Now, we use Claim 2 and Claim 3 to show that the size of $R$ (and consequently $U$) is at most seven, i.e., $w\le 7$.
Consider the lines $l_{a_{i+1}}^0$, $l_{a_{i+1}}^{60}$, $l_{a_{i+1}}^{120}$, and $l_{a_{i}}^{120}$ as shown in Figure~\ref{hamiltonicity-fig1}. Let $l_1$ and $l_2$ be the rays starting at the corners of $t(a_i, a_{i+1})$ opposite to $a_{i+1}$ and parallel to $l_{a_{i+1}}^0$ and $l_{a_{i+1}}^{60}$ respectively, as shown in Figure~\ref{hamiltonicity-fig1}. These lines and rays, partition the plane into 12 regions. We will show that each of the regions $D_1$, $D_2$, $D_3$, $D_4$, $C_1$, $C_2$, and $B=B_1\cup B_2$ contains at most one point of $R$, and the other regions do not contain any point of $R$. Consider the hexagon $X(a_{i+1},a_i)$. By Claim 2 and Observation~\ref{obs2}, no point of $R$ can be inside $X(a_{i+1},a_i)$. Moreover, no point of $R$ can be inside the cones $A_1$, $A_2$, and $A_3$, because if $r_j\in \{A_1 \cup A_2\cup A_3\}$, the (upward) triangle $t'(u_j, r_j)$ contains $a_{i+1}$. Then by Observation~\ref{obs2}, $t(r_j, a_{i+1}) < t(u_j, r_j)$; which contradicts Claim 2.

Now we show that each of the regions $D_1$, $D_2$, $D_3$ and $D_4$ contains at most one point of $R$. Consider the region $D_1$; by similar reasoning we can prove this claim for $D_2$, $D_3$, and $D_4$. Using contradiction, let $r_j$ and $r_k$ be two points in $D_1$, and w.l.o.g. assume that $r_j$ is the farthest to $l_{a_{i+1}}^{60}$. Then $r_k$ can lie inside any of the cones $C_{r_j}^1$, $C_{r_j}^5$, and $C_{r_j}^6$ (but not in $X$). If $r_k \in C_{r_j}^1$, then $t'(r_j, r_k)$ is smaller than $t'(a_i, a_{i+1})$ which means that $t(r_j, r_k)< t(a_i,a_{i+1})$. If $r_k \in C_{r_j}^5$, then $t'(u_j,r_j)$ contains $r_k$, that is $t(r_j, r_k)< t(u_j, r_j)$. If $r_k \in C_{r_j}^6$, then $t(u_j,r_j)$ contains $r_k$, that is $t(r_j, r_k)< t(u_j, r_j)$. All cases contradict Claim 3. 

Now consider the region $C_1$ (or its symmetric region $C_2$) and by contradiction assume that it contains two points $r_j$ and $r_k$. Let $r_j$ be the farthest from $l_{a_{i+1}}^{0}$. It is obvious that the $t'(u_j, r_j)$ contains $r_k$, that is $t(r_j, r_k)< t(u_j, r_j)$; which contradicts Claim 3. 

Now consider the region $B=B_1\cup B_2$. If both $r_j$ and $r_k$ belong to $B_2$, then $t'(r_j, r_k)$ is smaller that $t(a_i, a_{i+1})$. If $r_j\in B_1$ and $r_k\in B_2$, then $t'(u_j,r_j)$ contains $r_k$, and hence $t(r_j, r_k)<t(u_j, r_j)$. If both $r_j$ and $r_k$ belong to $B_1$, let $r_j$ be the farthest from $l_{a_i}^{120}$. Clearly, $t(u_j, r_j)$ contains $r_k$ and hence $t(r_j, r_k)<t(u_j, r_j)$. All cases contradict Claim 3.

Therefore, any of the regions $D_1$, $D_2$, $D_3$, $D_4$, $C_1$, $C_2$, and $B=B_1\cup B_2$ contains at most one point of $R$. Thus, $w \le 7$, and $t(a_i, a_{i+1})$ contains at most 7 points of $P$. Therefore, $e=(a_i, a_{i+1})$ is an edge of \kTD{7}{}.
\end{proof}

As a direct consequence of Theorem~\ref{hamiltonicity-thr} we have shown that:
\begin{corollary}
 \kTD{7}{} is Hamiltonian.
\end{corollary}

\begin{figure}[htb]
  \centering
  \includegraphics[width=.7\columnwidth]{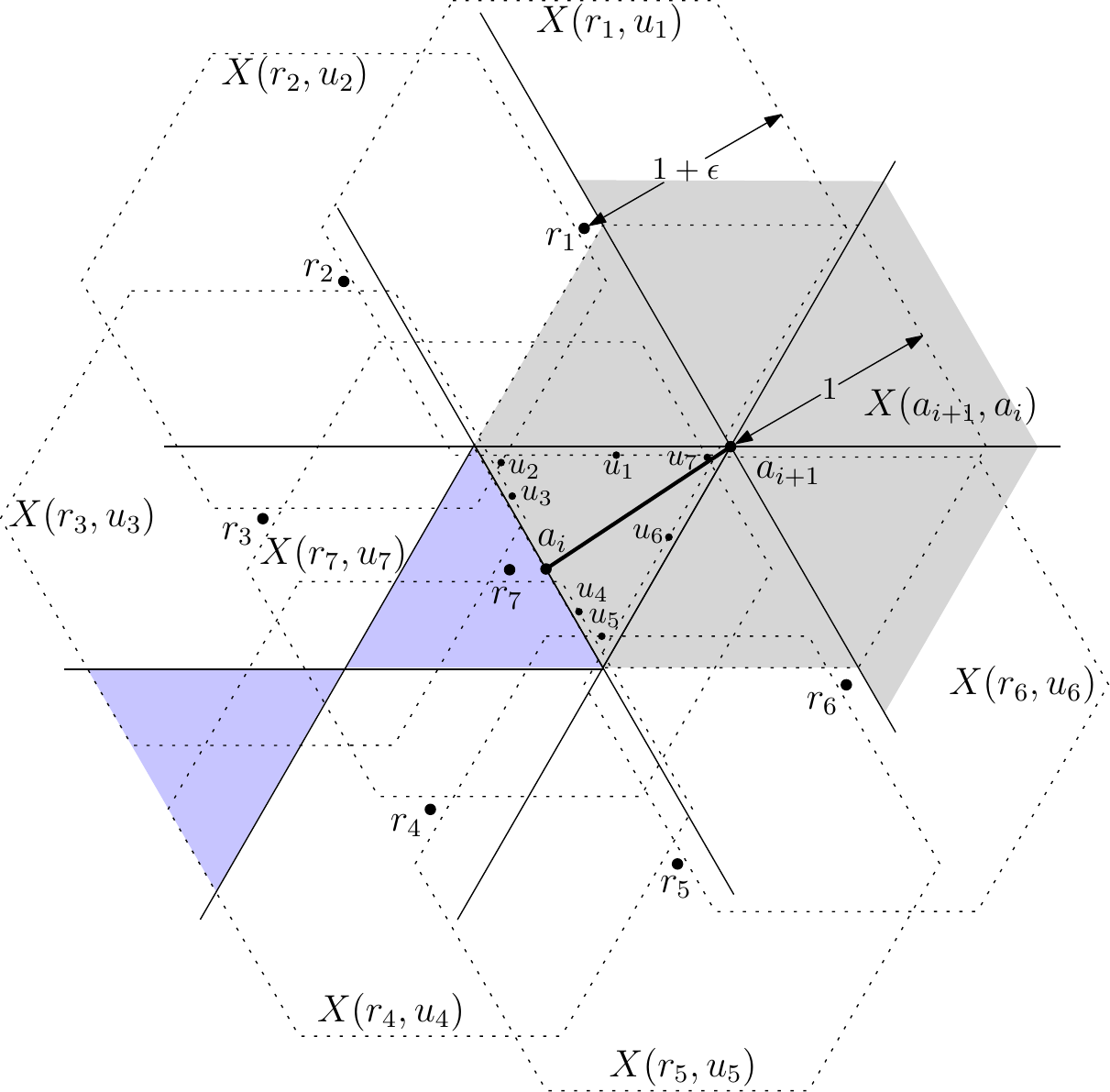}
 \caption{$t(a_i,a_{i+1})$ contains 7 points while the conditions in the proof of Theorem~\ref{hamiltonicity-thr} hold.}
  \label{hamiltonicity-fig2}
\end{figure}

An interesting question is to determine if \kTD{k}{} contains a bottleneck Hamiltonian cycle for $k<7$.
Figure~\ref{hamiltonicity-fig2} shows a configuration where $t(a_i, a_{i+1})$ contains 7 points while the conditions of Claim 1, Claim 2, and Claim 3 in the proof of Theorem~\ref{hamiltonicity-thr} hold. In Figure~\ref{hamiltonicity-fig2}, $d(a_i, a_{i+1})=1$, $d(r_i, u_i)= 1+\epsilon$, $d(r_i, r_j)>1+\epsilon$, $d(r_i, a_{i+1})> 1+\epsilon$ for $i,j=1,\dots 7$ and $i\neq j$.

\begin{figure}[htb]
  \centering
  \includegraphics[width=.8\columnwidth]{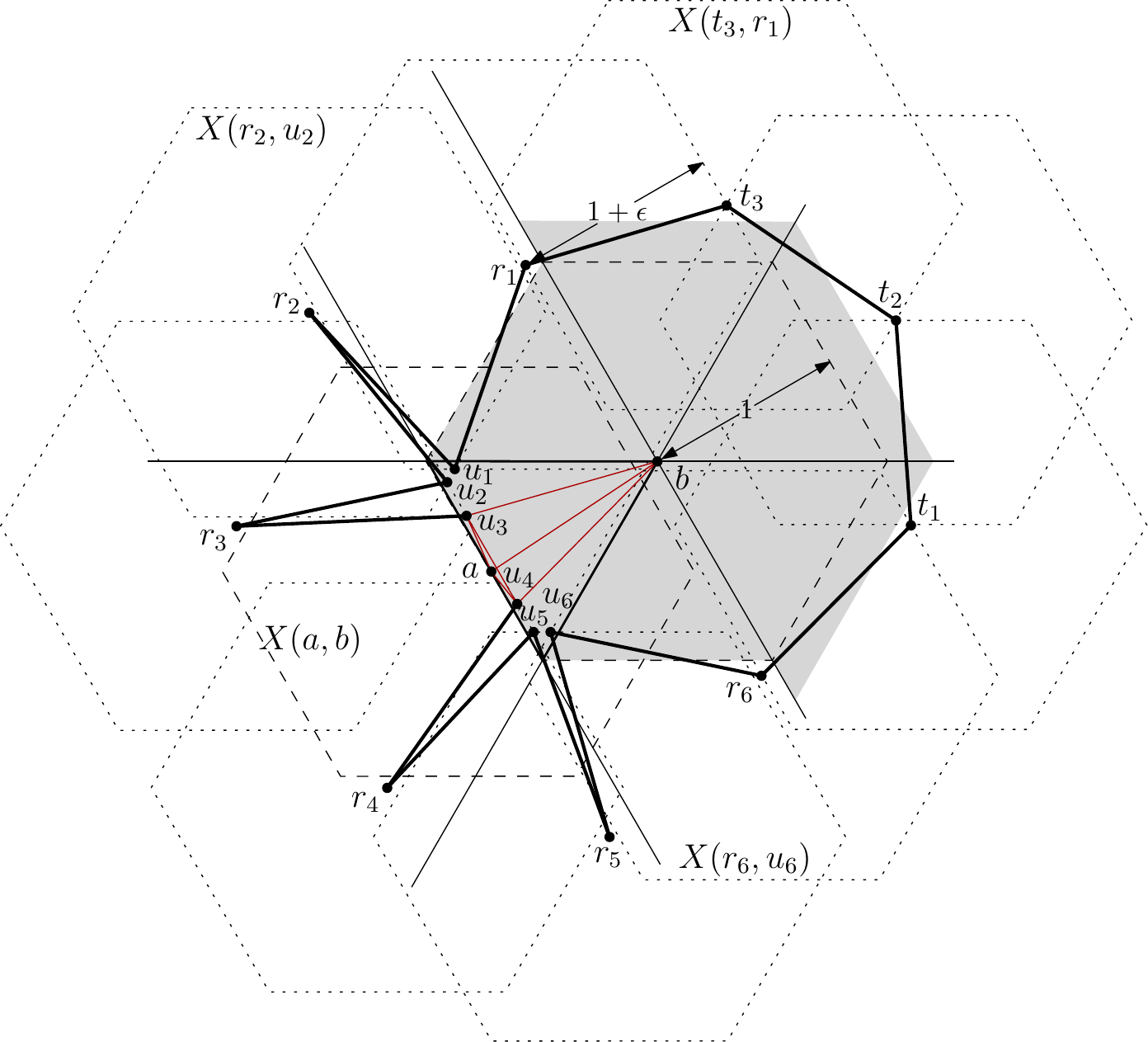}
 \caption{The points $\{r_1,\dots,r_6,t_1,t_2,t_3\}$ are connected to their first and second closest point (the bold edges). The edge $(a,b)$ should be in any bottleneck Hamiltonian cycle, while $t(a,b)$ contains 6 points.}
  \label{hamiltonicity-fig3}
\end{figure}

Figure~\ref{hamiltonicity-fig3} shows a configuration of $P$ with 17 points such that \kTD{5}{} does not contain a bottleneck Hamiltonian cycle. In Figure~\ref{hamiltonicity-fig3}, $d(a,b)=1$ and $t(a,b)$ contains 6 points $U=\{u_1, \dots, u_6\}$. In addition $d(r_i, u_i)= 1+\epsilon$, $d(r_i, r_j)>1+\epsilon$, $d(r_i, b)> 1+\epsilon$ for $i,j=1,\dots 6$ and $i\neq j$. Let $R=\{t_1, t_2,t_3, r_1,\dots, r_6\}$. The dashed hexagons are centered at $a$ and $b$ and have diameter 1. The dotted hexagons are centered at vertices in $R$ and have diameter $1+\epsilon$. Each point in $R$ is connected to its first and second closest points by edges of length $1+\epsilon$ (the bold edges). Let $B$ be the set of these edges. Let $H$ be a cycle formed by $B\cup\{(u_3,b),(b,a),(a,u_5)\}$, i.e., $H=(u_4,r_4,u_5,r_5,u_6,r_6,t_1,t_2,t_3,r_1,u_1,r_2,u_2,r_3,u_3,a,b,u_4)$. It is obvious that $H$ is a Hamiltonian cycle for $P$ and $\lambda(H)=1+\epsilon$. Thus, the bottleneck of any bottleneck Hamiltonian cycle for $P$ is at most $1+\epsilon$. We will show that any bottleneck Hamiltonian cycle for $P$ contains the edge $(a,b)$ which does not belong to \kTD{5}{}. By contradiction, let $H^*$ be a bottleneck Hamiltonian cycle which does not contain $(a,b)$. In $H^*$, $b$ is connected to two vertices $b_l$ and $b_r$, where $b_l\neq a$ and $b_r\neq a$. Since the distance between $b$ and any vertex in $R$ is strictly bigger than $1+\epsilon$ and $\lambda(H^*)\le 1+\epsilon$, $b_l\notin R$ and $b_r\notin R$. Thus $b_l$ and $b_r$ belong to $U$. Let $U'=\{u_1,u_2,u_5,u_6\}$. Consider two cases:

\begin{itemize}
 \item $b_l\in U'$ or $b_r\in U'$. W.l.o.g. assume that $b_l\in U'$ and $b_l=u_1$. Since $u_1$ is the first/second closest point of $r_1$ and $r_2$, in $H^*$ one of $r_1$ and $r_2$ must be connected by an edge $e$ to a point that is farther than its second closet point; $e$ has length strictly greater than $1+\epsilon$.
 \item $b_l\notin U'$ and $b_r\notin U'$. Thus, both $b_l$ and $b_r$ belong to $\{u_3,u_4\}$. That is, in $H^*$, $a$ should be connected to a point $c$ where $c\in R\cup U'$. If $c\in R$ then the edge $(a,c)$ has length more than $1+\epsilon$. If $c\in U'$, w.l.o.g. assume $c=u_1$; by the same argument as in the previous case, one of $r_1$ and $r_2$ must be connected by an edge $e$ to a point that is farther than its second closet point; $e$ has length strictly greater than $1+\epsilon$.
\end{itemize}

Since $e\in H^*$, both cases contradicts that $\lambda(H^*)\le 1+\epsilon$. Therefore, every bottleneck Hamiltonian cycle contains edge $(a,b)$. Since $(a,b)$ is not an edge in \kTD{5}{}, a bottleneck Hamiltonian cycle of $P$ is not contained in \kTD{5}{}.  

\section{Perfect Matching Admissibility}
\label{matching}
In this section we consider the matching problem in higher-order triangular-distance Delaunay graphs. In Subsection~\ref{bottleneck-matching-section} we show that \kTD{6}{} contains a bottleneck perfect matching. We also show that for some point sets $P$, \kTD{5}{} does not contain any bottleneck perfect matching. In Subsection~\ref{matching2} we prove that every \kTD{2}{} has a perfect matching when $P$ has an even number of points, and \kTD{1}{} contains a matching of size at least $\frac{2(n-1)}{5}$.

\subsection{Bottleneck Perfect Matching}
\label{bottleneck-matching-section}
\begin{theorem}
\label{matching-thr}
 For a set $P$ of an even number of points, \kTD{6}{} contains a bottleneck perfect matching.
\end{theorem}

\begin{proof}
Let $\mathcal{M}$ be the set of all perfect matchings through the points of $P$. Define a total order on the elements of $\mathcal{M}$ by their weight sequence. If two elements have exactly the same weight sequence, break ties arbitrarily to get a total order.
Let $M^* = \{(a_1, b_1),\dots, (a_{\frac{n}{2}}, b_{\frac{n}{2}})\}$ be a perfect matching in $\mathcal{M}$ with minimal weight sequence. It is obvious that $M^*$ is a bottleneck perfect matching for $P$. We will show that all edges of $M^*$ are in \kTD{6}{}. Consider any edge $e = (a_i, b_i)$ in $M^*$ and its corresponding triangle $t(a_i,b_i)$.

\begin{figure}[htb]
  \centering
  \includegraphics[width=.6\columnwidth]{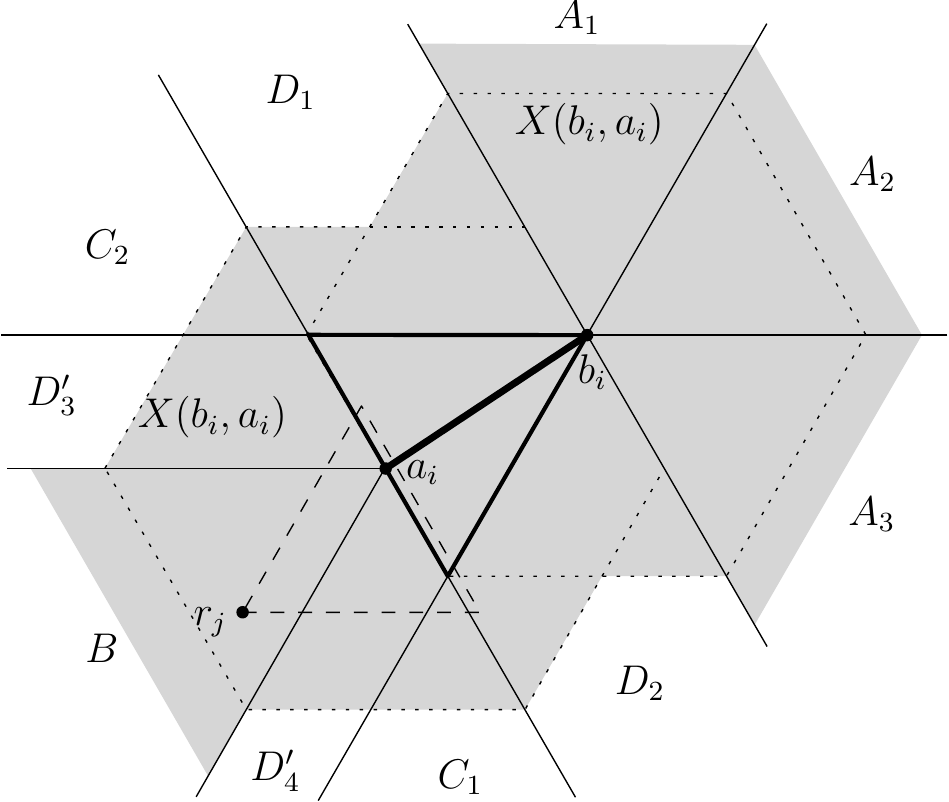}
 \caption{Proof of Theorem~\ref{matching-thr}.}
  \label{matching-fig1}
\end{figure}

{\em Claim 1}: None of the edges of $M^*$ can be inside $t(a_i,b_i)$. Suppose there is an edge $f=(a_j, b_j)$ inside $t(a_i,b_i)$. Let $M$ be a perfect matching obtained from $M^*$ by deleting $\{e, f\}$, and adding $\{(a_i, a_j), (b_i, b_j)\}$. By Observation~\ref{obs1}, the two new edges are smaller than the old ones. Thus, $\WS{M}<\WS{M^*}$ which contradicts the minimality of $M^*$.

Therefore, we may assume that no edge of $M^*$ lies completely inside $t(a_i,b_i)$. Suppose there are $w$ points of $P$ inside $t(a_i,b_i)$. Let $U = u_1, u_2,\dots, u_w$ represent the points inside $t(a_i, b_i)$, and $R=r_1, r_2,\dots, r_w$ represent the points where $(r_i,u_i)\in M^*$. W.l.o.g. assume that $a_i\in C^4_{b_i}$, and $t(a_i,b_i)$ is anchored at $b_i$ as shown in Figure~\ref{matching-fig1}.  

{\em Claim 2}: For each $r_j\in R$, $\min\{t(r_j, a_i), t(r_j, b_i)\} \ge \max\{t(a_i, b_i), t(u_j, r_j)\}$. By a similar argument as in the proof of Claim 2 in Theorem \ref{hamiltonicity-thr} we can either match $r_j$ with $a_i$ or $b_i$ to obtain a smaller matching $M$; which is a contradiction.

{\em Claim 3}: For each pair $r_j$ and $r_k$ of points in $R$, $t(r_j, r_k)\ge \max\{t(a_i, b_i), t(r_j, u_j), t(r_k, u_k)\}$. The proof is similar to the proof of Claim 3 in Theorem \ref{hamiltonicity-thr}.

Consider Figure~\ref{matching-fig1} which partitions the plane into eleven regions. As a direct consequence of Claim 2, the hexagons $X(b_i, a_i)$ and $X(a_i, b_i)$ do not contain any point of $R$. By a similar argument as in the proof of Theorem \ref{hamiltonicity-thr}, the regions $A_1$, $A_2$, $A_3$ do not contain any point of $R$. In addition, the region $B$ does not contain any point $r_j$ of $R$, because otherwise $t'(r_j,u_j)$ contains $a_i$, that is $t(r_j, a_i)< t(u_j, r_j)$ which contradicts Claim 2. As shown in the proof of Theorem \ref{hamiltonicity-thr} each of the regions $D_1$, $D_2$, $D'_3$, $D'_4$, $C_1$, and $C_2$ contains at most one point of $R$ (note that $D'_3\subset D_3$ and $D'_4\subset D_4$). Thus, $w \le 6$, and $t(a_i, b_i)$ contains at most 6 points of $P$. Therefore, $e=(a_i, b_i)$ is an edge of \kTD{6}{}.
\end{proof}

As a direct consequence of Theorem~\ref{matching-thr} we have shown that:
\begin{corollary}
 For a set $P$ of even number of points, \kTD{6}{} has a perfect matching.
\end{corollary}

\begin{figure}[htb]
  \centering
  \includegraphics[width=.6\columnwidth]{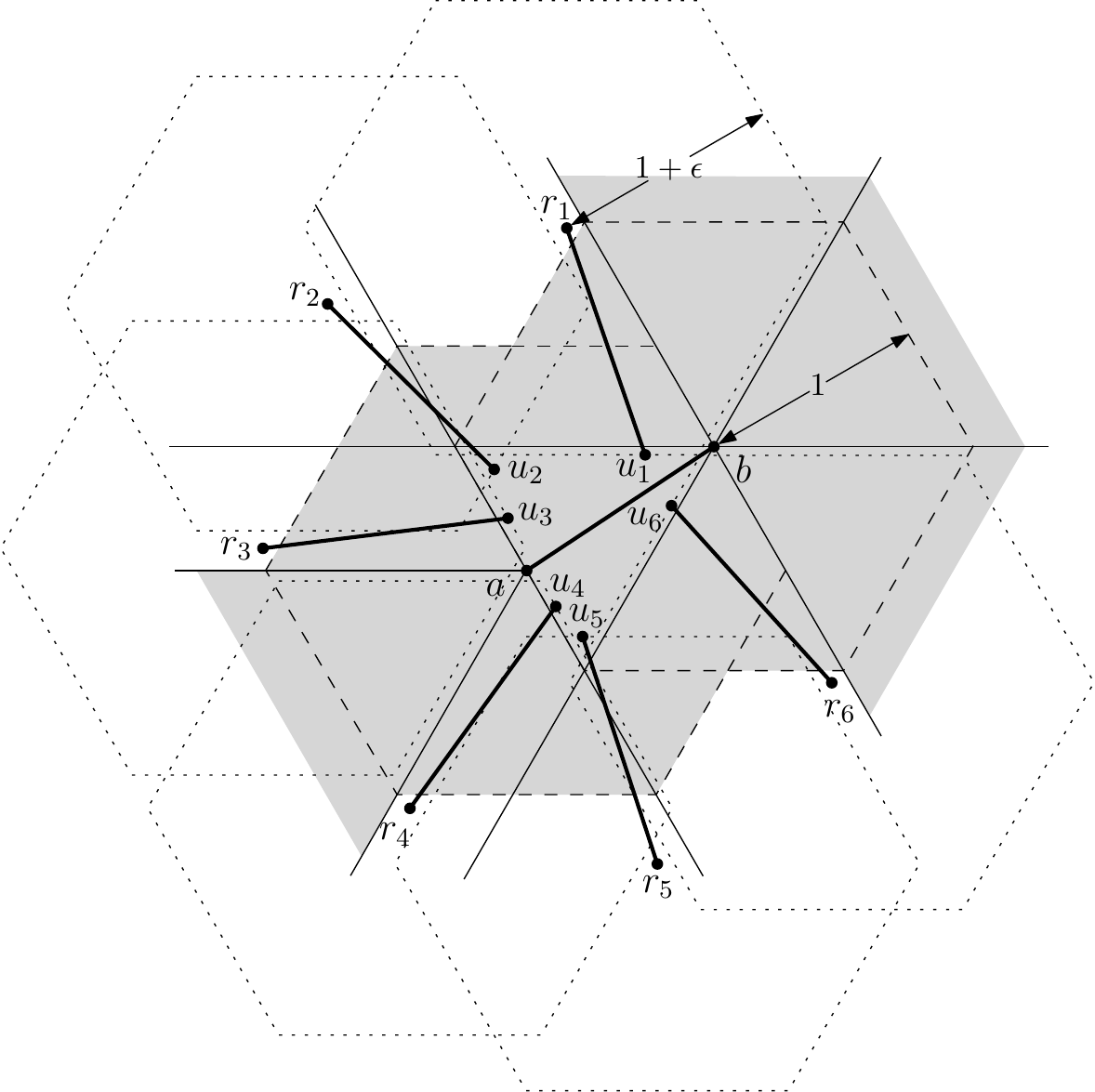}
 \caption{The points $\{r_1,\dots,r_6\}$ are matched to their closest point. The edge $(a, b)$ should be an edge in any bottleneck perfect matching, while $t(a, b)$ contains 6 points.}
  \label{matching-fig3}
\end{figure}

We show that the bound $k=6$ proved in Theorem~\ref{matching-thr} is tight. We will show that there are point sets $P$ such that \kTD{5}{} does not contain any bottleneck perfect matching.
Figure~\ref{matching-fig3} shows a configuration of $P$ with 14 points such that $d(a,b)=1$ and $t(a,b)$ contains six points $U=\{u_1, \dots, u_6\}$. In addition $d(r_i, u_i)= 1+\epsilon$, $d(r_i,x)>1+\epsilon$ where $x\neq u_i$, for $i=1,\dots 6$. Let $R=\{r_1,\dots, r_6\}$. In Figure~\ref{matching-fig3}, the dashed hexagons are centered at $a$ and $b$, each of diameter 1, and the dotted hexagons centered at vertices in $R$, each of diameter $1+\epsilon$. Consider a perfect matching $M=\{(a,b)\}\cup \{(r_i, u_i): i=1,\dots, 6\}$ where each point $r_i\in R$ is matched to its closest point $u_i$. It is obvious that $\lambda(M)=1+\epsilon$, and hence the bottleneck of any bottleneck perfect matching is at most $1+\epsilon$. We will show that any bottleneck perfect matching for $P$ contains the edge $(a,b)$ which does not belong to \kTD{5}{}. By contradiction, let $M^*$ be a bottleneck perfect matching which does not contain $(a,b)$. In $M^*$, $b$ is matched to a point $c\in R\cup U$. If $c \in R$, then $d(b,c)>1+\epsilon$. If $c\in U$, w.l.o.g. assume $c = u_1$. Thus, in $M^*$ the point $r_1$ is matched to a point $d$ where $d\neq u_1$. Since $u_1$ is the closest point to $r_1$ and $d(r_1,u_1)=1+\epsilon$, $d(r_1,d)>1+\epsilon$. Both cases contradicts that $\lambda(M^*)\le 1+\epsilon$. Therefore, every bottleneck perfect matching contains $(a,b)$. Since $(a,b)$ is not an edge in \kTD{5}{}, a bottleneck perfect matching of $P$ is not contained in \kTD{5}{}.  

\subsection{Perfect Matching}
\label{matching2}
In \cite{Babu2013} the authors proved a tight lower bound of $\lceil\frac{n-1}{3}\rceil$ on the size of a maximum matching in \kTD{0}{}. In this section we prove that \kTD{1}{} has a matching of size $\frac{2(n-1)}{5}$ and \kTD{2}{} has a perfect matching when $P$ has an even number of points.

\begin{figure}[htb]
  \centering
\setlength{\tabcolsep}{0in}
  $\begin{tabular}{cc}
 \multicolumn{1}{m{.5\columnwidth}}{\centering\includegraphics[width=.35\columnwidth]{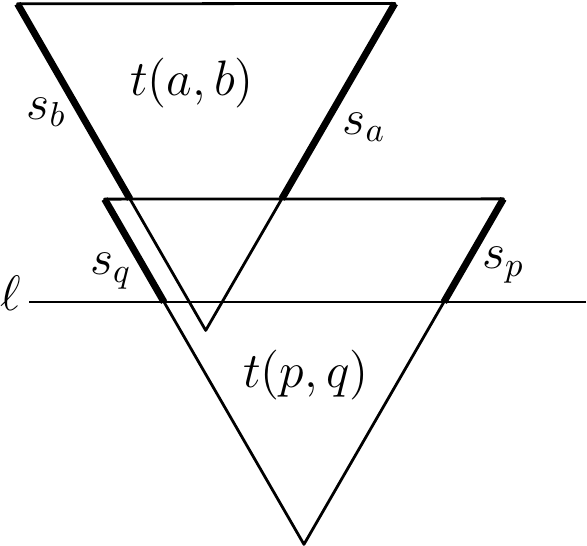}}
&\multicolumn{1}{m{.5\columnwidth}}{\centering\includegraphics[width=.35\columnwidth]{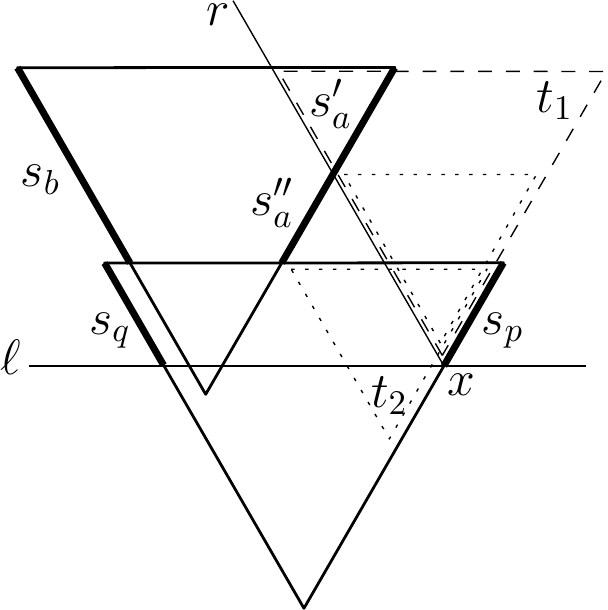}}\\
(a) & (b)
\end{tabular}$
  \caption{(a) Illustration of Lemma~\ref{triangle3}, and (b) proof of Lemma~\ref{triangle3}.}
\label{intersection-fig}
\end{figure}

For a triangle $t(a,b)$ through the points $a$ and $b$, let $top(a,b)$, $left(a,b)$, and $right(a,b)$ respectively denote the top, left, and right sides of $t(a,b)$. Refer to Figure~\ref{intersection-fig}(a) for the following lemma.
\begin{lemma}
\label{triangle3}
 Let $t(a,b)$ and $t(p,q)$ intersect a horizontal line $\ell$, and $t(a,b)$ intersects $top(p,q)$ in such a way that $t(p,q)$ contains the lowest corner of $t(a,b)$. If $a$ and $b$ lie above $top(p,q)$, and $p$ and $q$ lie above $\ell$, then, $\max\{t(a,p), t(b,q)\} < \max\{t(a,b), t(p,q)\}$.
\end{lemma}
\begin{proof}
Recall that $t(a,b)$ is the smallest downward triangle through $a$ and $b$. By Observation~\ref{side-point-obs} each side of $t(a,b)$ contains either $a$ or $b$. 
In Figure~\ref{intersection-fig}(a) the set of potential positions for point $a$ on the boundary of $t(a,b)$ is shown by the line segment $s_a$; and similarly by $s_b$, $s_p$, $s_q$ for $b$, $p$, $q$, respectively. We will show that $t(a,p)<\max\{t(a,b), t(p,q)\}$. By similar reasoning we can show that $t(b,q)<\max\{t(a,b), t(p,q)\}$. Let $x$ denote the intersection of $\ell$ and $right(p,q)$. Consider a ray $r$ initiated at $x$ and parallel to $left(p,q)$ which divides $s_a$ into (at most) two parts $s'_a$ and $s''_a$ as shown in Figure~\ref{intersection-fig}(b). Two cases may appear:

\begin{itemize}
 \item $a\in s'_a$. Let $t_1$ be a downward triangle anchored at $x$ which has its $top$ side on the line through $top(a,b)$ (the dashed triangle in Figure~\ref{intersection-fig}(b)). The top side of $t_1$ and $t(a,b)$ lie on the same horizontal line. The bottommost corner of $t_1$ is on $\ell$ while the bottommost corner of $t(a,b)$ is below $\ell$. Thus, $t_1<t(a,b)$. In addition, $t_1$ contains $s'_a$ and $s_p$, thus, for any two points $a\in s'_a$ and $p\in s_p$, $t(a,p)\le t_1$. Therefore, $t(a,p)< t(a,b)$.
 \item $a \in s''_a$. Let $t_2$ be a downward triangle anchored at the intersection of $right(a,b)$ and $top(p,q)$ which has one side on the line through $right(p,q)$ (the dotted triangle in Figure~\ref{intersection-fig}(b)). This triangle is contained in $t(p,q)$, and has $s_p$ on its right side. If we slide $t_2$ upward while its top-left corner remains on $s''_a$, the segment $s_p$ remains on the right side of $t_2$. Thus, any triangle connecting a point $a\in s''_a$ to a point $p\in s_p$ has the same size as $t_2$. That is, $t(a,p)=t_2<t(p,q)$. 
\end{itemize}

Therefore, we have $t(a,p)<\max\{t(a,b), t(p,q)\}$. By similar argument we conclude that $t(b,q)<\max\{t(a,b), t(p,q)\}$.  
\end{proof}
Let $\mathcal{P}=\{P_1, P_2,\dots\}$ be a partition of the points in $P$.
Let $G(\mathcal{P})$ be a complete graph with vertex set $\mathcal{P}$. For each edge $e=(P_i,P_j)$ in $G(\mathcal{P})$, let $w(e)$ be equal to the area of the smallest triangle between a point in $P_i$ and a point in $P_j$, i.e. $w(e)=\min\{t(a,b):a\in P_i, b\in P_j\}$. That is, the weight of an edge $e\in G(\mathcal{P})$ corresponds to the size of the smallest triangle $t(e)$ defined by the endpoints of $e$. Let $\mathcal{T}$ be a minimum spanning tree of $G(\mathcal{P})$. Let $T$ be the set of triangles corresponding to the edges of $\mathcal{T}$, i.e. $T=\{t(e): e\in \mathcal{T}\}$. 

\begin{lemma}
 \label{empty-triangle-lemma}
The interior of any triangle in $T$ does not contain any point of $P$.
\end{lemma}
\begin{proof}
  By contradiction, suppose there is a triangle $\tau\in T$ which contains a point $c\in P$. Let $e=(P_i,P_j)$ be the edge in $\mathcal{T}$ which corresponds to $\tau$. Let $a$ and $b$ respectively be the points in $P_i$ and $P_j$ which define $\tau$, i.e. $\tau=t(a,b)$ and $w(e)=t(a,b)$. Three cases arise: (i) $c\in P_i$, (ii) $c\in P_j$, (iii) $c\in P_l$ where $l\neq i$ and $l\neq j$. In case (i) the triangle $t(c,b)$ between $c\in P_i$ and $b\in P_j$ is smaller than $t(a,b)$; contradicts that $w(e)=t(a,b)$ in $G(\mathcal{P})$.  In case (ii) the triangle $t(a,c)$ between $a\in P_i$ and $c\in P_j$ is smaller than $t(a,b)$; contradicts that $w(e)=t(a,b)$ in $G(\mathcal{P})$. In case (iii) the triangle $t(a,c)$ (resp. $t(c,b)$) between $P_i$ and $P_l$ (resp. $P_l$ and $P_j$) is smaller than $t(a,b)$; contradicts that $e$ is an edge in $\mathcal{T}$. 
\end{proof}

\begin{lemma}
 \label{intersection-lemma}
Each point in the plane can be in the interior of at most three triangles in $T$. 
\end{lemma}

\begin{proof}
For each $t(a,b)\in T$, the sides $top(a,b)$, $right(a,b)$, and $left(a,b)$ contains at least one of $a$ and $b$. In addition, by Lemma~\ref{empty-triangle-lemma}, $t(a,b)$ does not contain any point of $P$ in its interior. Thus, none of $top(a,b)$, $right(a,b)$, and $left(a,b)$ is completely inside the other triangles. Therefore, the only possible way that two triangles $t(a,b)$ and $t(p,q)$ can share a point is that one triangle, say $t(p,q)$, contains a corner of $t(a,b)$ in such a way that $a$ and $b$ are outside $t(p,q)$. In other words $t(a,b)$ intersects $t(p,q)$ through one of the sides $top(p,q)$, $right(p,q)$, or $left(p,q)$. If $t(a,b)$ intersects $t(p,q)$ through a direction $d\in \{top, right, left\}$ we say that $t(p,q)\prec_{d} t(a,b)$. 

By contradiction, suppose there is a point $c$ in the plane which is inside four triangles $\{t_1,t_2,t_3,t_4\}\subseteq T$. Out of these four, either (i) three of them are like $t_i\prec_d t_j \prec_d t_k$ or (ii) there is a triangle $t_l$ such that $t_l\prec_{top} t_i, t_l\prec_{right} t_j, t_l\prec_{left} t_k$, where $1\le i,j, k,l\le 4$ and $i\neq j \neq k \neq l$. Figure~\ref{configuration-fig} shows the two possible configurations (note that all other configurations obtained by changing the indices of triangles and/or the direction are symmetric to Figure~\ref{configuration-fig}(a) or Figure~\ref{configuration-fig}(b)).
\begin{figure}[htb]
  \centering
\setlength{\tabcolsep}{0in}
  $\begin{tabular}{cc}
 \multicolumn{1}{m{.5\columnwidth}}{\centering\includegraphics[width=.25\columnwidth]{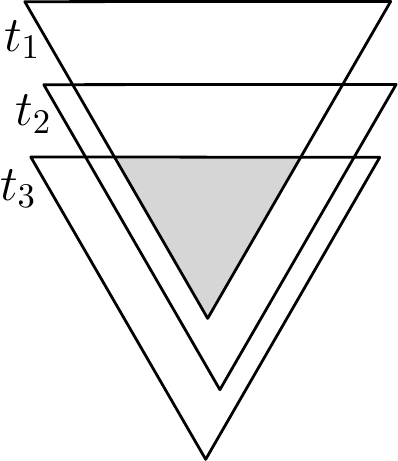}}
&\multicolumn{1}{m{.5\columnwidth}}{\centering\includegraphics[width=.35\columnwidth]{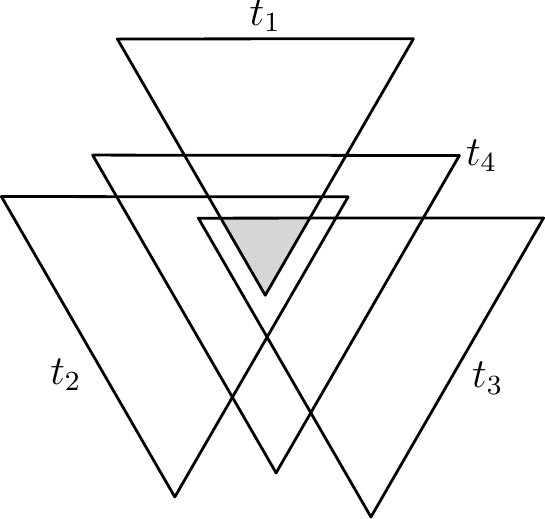}}\\
(a) & (b)
\end{tabular}$
  \caption{Two possible configurations: (a) $t_3 \prec_{top} t_2 \prec_{top} t_1$, (b) $t_4\prec_{top} t_1, t_4\prec_{left} t_2, t_4\prec_{right} t_3$.}
\label{configuration-fig}
\end{figure}

 Recall that each of $t_1,t_2,t_3,t_4$ corresponds to an edge in $\mathcal{T}$. In the configuration of Figure~\ref{configuration-fig}(a) consider $t_1$, $t_2$, and $top(t_3)$ which is shown in more detail in Figure~\ref{matching3-fig}(a). Suppose $t_1$ (resp. $t_2$) is defined by points $a$ and $b$ (resp. $p$ and $q$). By Lemma~\ref{empty-triangle-lemma}, $p$ and $q$ are above $top(t_3)$, $a$ and $b$ are above $top(t_2)$. By Lemma~\ref{triangle3}, $\max\{t(a,p),t(b,q)\} <\max\{t(a,b), t(p,q)\}$. This contradicts the fact that both of the edges representing $t(a,b)$ and $t(p,q)$ are in $\mathcal{T}$, because by replacing $\max\{t(a,b), t(p,q)\}$ with $t(a,p)$ or $t(b,q)$, we obtain a tree $\mathcal{T'}$ which is smaller than $\mathcal{T}$. In the configuration of Figure~\ref{configuration-fig}(b), consider all pairs of potential positions for two points defining $t_4$ which is shown in more detail in Figure~\ref{matching3-fig}(b). The pairs of potential positions on the boundary of $t_4$ are shown in red, green, and orange. Consider the red pair, and look at $t_2$, $t_4$, and $left(t_1)$. By Lemma~\ref{triangle3} and the same reasoning as for the previous configuration, we obtain a smaller tree $\mathcal{T'}$;  which contradicts the minimality of $\mathcal{T}$. By symmetry, the green and orange pairs lead to a contradiction.
Therefore, all configurations are invalid; which proves the lemma.

\begin{figure}[htb]
  \centering
\setlength{\tabcolsep}{0in}
  $\begin{tabular}{cc}
 \multicolumn{1}{m{.5\columnwidth}}{\centering\includegraphics[width=.35\columnwidth]{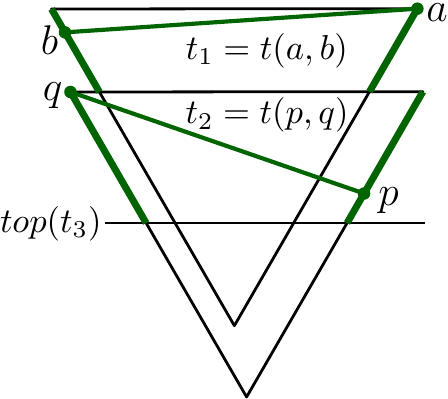}}
&\multicolumn{1}{m{.5\columnwidth}}{\centering\includegraphics[width=.4\columnwidth]{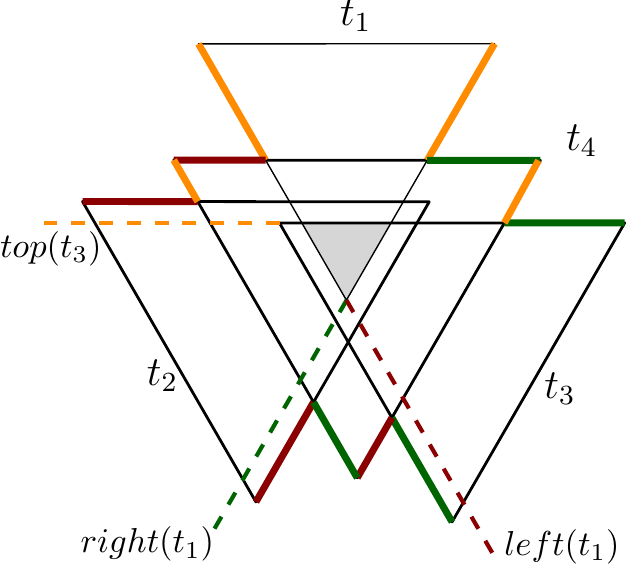}}\\
(a) & (b)
\end{tabular}$
  \caption{Illustration of Lemma~\ref{intersection-lemma}.}
\label{matching3-fig}
\end{figure}
\end{proof}

Our results in this section are based on Lemma~\ref{empty-triangle-lemma}, Lemma~\ref{intersection-lemma} and the two theorems by Tutte~\cite{Tutte1947} and Berge~\cite{Berge1958}. 

Now we prove that \kTD{2}{} has a perfect matching.

\begin{theorem}
 \label{mt-thr}
For a set $P$ of an even number of points, \kTD{2}{} has a perfect matching.
\end{theorem}
\begin{proof}
First we show that by removing a set $K$ of $k$ points from \kTD{2}{}, at most $k+1$ components are generated. Then we show that at least one of these components must be even. Finally by Theorem~\ref{Tutte} we conclude that \kTD{2}{} has a perfect matching.

Let $K$ be a set of $k$ vertices removed from \kTD{2}{}, and let $\mathcal{C}=\{C_1, \dots, C_{m(k)}\}$ be the resulting $m(k)$ components, where $m$ is a function depending on $k$. Actually $\mathcal{C}=\text{\kTD{2}{}}-K$ and $\mathcal{P}=\{V(C_1),\dots, V(C_{m(k)})\}$ is a partition of the vertices in $P\setminus K$. 

{\bf\em  Claim 1.} $m(k)\le k+1$. Let $G(\mathcal{P})$ be a complete graph with vertex set $\mathcal{P}$ which is constructed as described above. Let $\mathcal{T}$ be a minimum spanning tree of $G(\mathcal{P})$ and let $T$ be the set of triangles corresponding to the edges of $\mathcal{T}$. It is obvious that $\mathcal{T}$ contains $m(k)-1$ edges and hence $|T|=m(k)-1$. Let $F=\{(p,t):p\in K, t\in T, p\in t\}$ be the set of all (point, triangle) pairs where $p\in K$, $t\in T$, and $p$ is inside $t$. By Lemma~\ref{intersection-lemma} each point in $K$ can be inside at most three triangles in $T$. Thus, $|F|\le 3\cdot|K|$.
Now we show that each triangle in $T$ contains at least three points of $K$.  
Consider any triangle $\tau\in T$. Let $e=(V(C_i),V(C_j))$ be the edge of $\mathcal{T}$ which is corresponding to $\tau$, and let $a\in V(C_i)$ and $b\in V(C_j)$ be the points defining $\tau$. By Lemma~\ref{empty-triangle-lemma}, $\tau$ does not contain any point of $P\setminus K$ in its interior. Therefore, $\tau$ contains at least three points of $K$, because otherwise $(a,b)$ is an edge in \kTD{2}{} which contradicts the fact that $a$ and $b$ belong to different components in $\mathcal{C}$. Thus, each triangle in $T$ contains at least three points of $K$ in its interior. That is, $3\cdot|T|\le|F|$. Therefore, $3(m(k)-1)\le |F|\le 3k$, and hence $m(k)\le k+1$.

{\bf \em Claim 2}: $o(\mathcal{C})\le k$. By Claim 1, $|\mathcal{C}|=m(k)\le k+1$. If $|\mathcal{C}|\le k$, then $o(\mathcal{C})\le k$. Assume that $|\mathcal{C}|=k+1$. Since $P=K\cup \{\bigcup^{k+1}_{i=1}{V(C_i)}\}$, the total number of vertices of $P$ can be defined as $n=k+\sum_{i=1}^{k+1}{|V(C_i)|}$. Consider two cases where (i) $k$ is odd, (ii) $k$ is even. In both cases if all the components in $\mathcal{C}$ are odd, then $n$ is odd; contradicts our assumption that $P$ has an even number of vertices. Thus, $\mathcal{C}$ contains at least one even component, which implies that $o(\mathcal{C})\le k$.

Finally, by Claim 2 and Theorem~\ref{Tutte}, we conclude that \kTD{2}{} has a perfect matching.
\end{proof}

\begin{theorem}
\label{matching-1TD}
 For every set $P$ of points, \kTD{1}{} has a matching of size $\frac{2(n-1)}{5}$.
\end{theorem}

\begin{proof}
Let $K$ be a set of $k$ vertices removed from \kTD{1}{}, and let $\mathcal{C}=\{C_1, \dots, C_{m(k)}\}$ be the resulting $m(k)$ components. Actually $\mathcal{C}=\text{\kTD{1}{}}-K$ and $\mathcal{P}=\{V(C_1),\dots, V(C_{m(k)})\}$ is a partition of the vertices in $P\setminus K$. Note that $o(\mathcal{C})\le m(k)$.
Let $M^*$ be a maximum matching in \kTD{1}{}. By Theorem~\ref{Berge}, 

\begin{align}
\label{align0}
|M^*|&= \frac{1}{2}(n-\text{def}(\text{\kTD{1}{}})),
\end{align}

where

\begin{align}
\label{align1}
\text{def}(\text{\kTD{1}{}})&= \max\limits_{K\subseteq P}(o(\mathcal{C})-|K|)\nonumber\\
& \le \max\limits_{K\subseteq P}(|\mathcal{C}|-|K|)\nonumber\\
& = \max\limits_{0\le k\le n}(m(k)-k).
\end{align}
Define $G(\mathcal{P})$, $\mathcal{T}$, $T$, and $F$ as in the proof of Theorem~\ref{mt-thr}. By Lemma~\ref{intersection-lemma}, $|F|\le 3\cdot|K|$.
By the same reasoning as in the proof of Theorem~\ref{mt-thr}, each triangle in $T$ has at least two points of $K$ in its interior. Thus, $2\cdot|T|\le|F|$. Therefore, $2(m(k)-1)\le |F| \le 3k$, and hence

\begin{equation}
\label{ineq1}
 m(k)\le\frac{3k}{2}+1.
\end{equation} 

In addition, $k+m(k)=|K|+|\mathcal{C}|\le |P|=n$, and hence

\begin{equation}
\label{ineq2}               
m(k)\le n-k.
\end{equation}

By Inequalities~(\ref{ineq1}) and ~(\ref{ineq2}), 

\begin{equation}
\label{ineq3}               
m(k)\le \min\{\frac{3k}{2}+1, n-k\}.
\end{equation}

Thus, by (\ref{align1}) and (\ref{ineq3})

\begin{align}
\label{align2}
\text{def}(\text{\kTD{1}{}})&\le \max\limits_{0\le k\le n}(m(k)-k)\nonumber\\
&\le \max\limits_{0\le k\le n}\{\min\{\frac{3k}{2}+1, n-k\}-k\}\nonumber\\
&= \max\limits_{0\le k\le n}\{\min\{\frac{k}{2}+1, n-2k\}\}\nonumber\\
&= \frac{n+4}{5},
\end{align}

where the last equation is achieved by setting $\frac{k}{2}+1$ equal to $n-2k$, which implies $k=\frac{2(n-1)}{5}$. Finally by substituting (\ref{align2}) in Equation (\ref{align0}) we have
$$
|M^*|\ge \frac{2(n-1)}{5}.
$$
\end{proof}
\section{Blocking TD-Delaunay graphs}
\label{blocking-section}
In this section we consider the problem of blocking TD-Delaunay graphs. Let $P$ be a set of $n$ points in the plane such that no pair of points of $P$ is collinear in the $l^{0}$, $l^{60}$, and $l^{120}$ directions. Recall that a point set $K$ blocks \kTD{k}{$(P)$} if in \kTD{k}{$(P\cup K)$} there is no edge connecting two points in $P$. That is, $P$ is an independent set in \kTD{k}{$(P\cup K)$}.

\begin{theorem}
\label{blocking-thr1}
  At least $\lceil\frac{(k+1)(n-1)}{3}\rceil$ points are necessary to block \kTD{k}{$(P)$}.
\end{theorem}
\begin{proof}
Let $K$ be a set of $m$ points which blocks \kTD{k}{$(P)$}. Let $G(\mathcal{P})$ be a complete graph with vertex set $\mathcal{P}=P$. Let $\mathcal{T}$ be a minimum spanning tree of $G(\mathcal{P})$ and let $T$ be the set of triangles corresponding to the edges of $\mathcal{T}$. It is obvious that $|T|=n-1$. By Lemma~\ref{empty-triangle-lemma} the triangles in $T$ are empty, thus, the edges of $\mathcal{T}$ belong to any \kTD{k}{$(P)$} where $k\ge0$. To block each edge, corresponding to a triangle in $T$, at least $k+1$ points are necessary. By Lemma~\ref{intersection-lemma} each point in $K$ can lie in at most three triangles of $T$. Therefore, $m\ge\lceil\frac{(k+1)(n-1)}{3}\rceil$, which implies that at least $\lceil\frac{(k+1)(n-1)}{3}\rceil$ points are necessary to block all the edges of $\mathcal{T}$ and hence \kTD{k}{$(P)$}.
\end{proof}
\begin{figure}[htb]
  \centering
\setlength{\tabcolsep}{0in}
  $\begin{tabular}{cc}
 \multicolumn{1}{m{.5\columnwidth}}{\centering\includegraphics[width=.45\columnwidth]{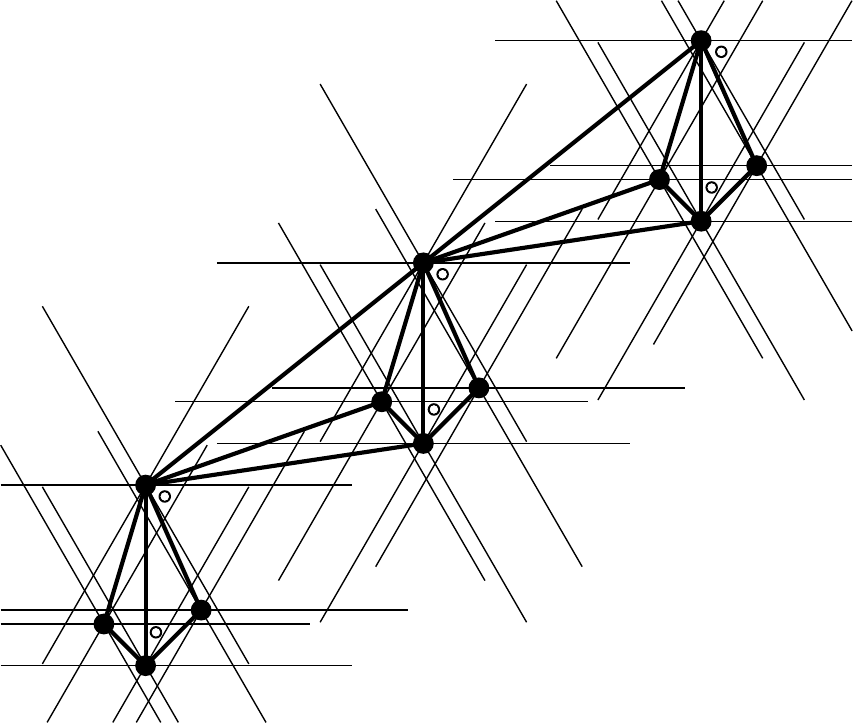}}
&\multicolumn{1}{m{.5\columnwidth}}{\centering\includegraphics[width=.4\columnwidth]{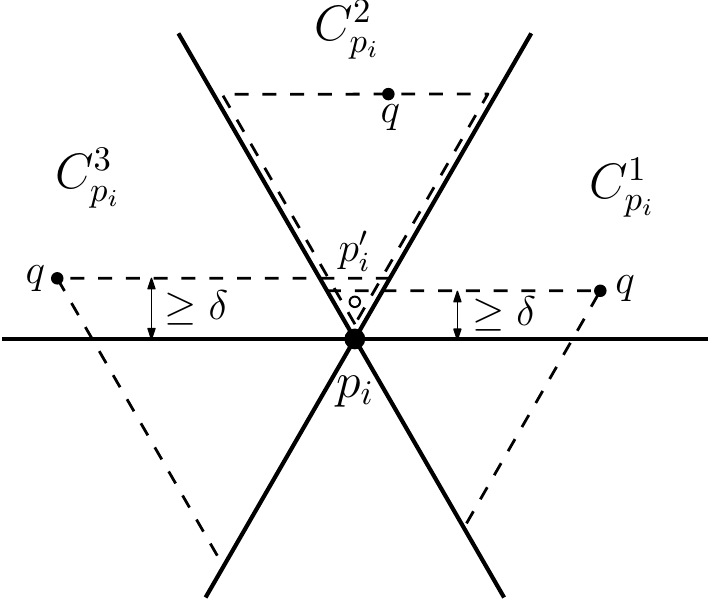}}\\
(a) & (b)
\end{tabular}$
  \caption{(a) a \kTD{0}{} graph which is shown in bold edges is blocked by $\lceil\frac{n-1}{2}\rceil$ white points, (b) $p'_i$ blocks all the edges connecting $p_i$ to the vertices above $l^0_{p_i}$.}
\label{blocking-fig}
\end{figure}

Theorem~\ref{blocking-thr1} gives a lower bound on the number of points that are necessary to block a TD-Delaunay graph. By this theorem, at least $\lceil\frac{n-1}{3}\rceil$, $\lceil\frac{2(n-1)}{3}\rceil$, $n-1$ points are necessary to block $0\text{-}$, $1\text{-}$, \kTD{2}{$(P)$} respectively. Now we introduce another formula which gives a better lower bound for \kTD{0}{}. For a point set $P$, let $\nu_k(P)$ and $\alpha_k(P)$ respectively denote the size of a maximum matching and a maximum independent set in \kTD{k}{$(P)$}. For every edge in the maximum matching, at most one of its endpoints can be in the maximum independent set. Thus,
\begin{equation}
\label{matching-independent}
 \alpha_k(P)\le |P| - \nu_k(P).
\end{equation}
Let $K$ be a set of $m$ points which blocks \kTD{k}{$(P)$}. By definition there is no edge between points of $P$ in \kTD{k}{$(P\cup K)$}. That is, $P$ is an independent set in \kTD{k}{$(P\cup K)$}. Thus, 
\begin{equation}
\label{ineq4}
 n\le \alpha_k(P\cup K).
\end{equation}
By (\ref{matching-independent}) and (\ref{ineq4}) we have
\begin{equation}
\label{ineq5}
 n\le \alpha_k(P\cup K)\le (n+m)-\nu_k(P\cup K).
\end{equation}
\begin{theorem}
\label{blocking-0TD-thr}
  At least $\lceil\frac{n-1}{2}\rceil$ points are necessary to block \kTD{0}{$(P)$}.
\end{theorem}
\begin{proof}
Let $K$ be a set of $m$ points which blocks \kTD{k}{$(P)$}. Consider \kTD{0}{$(P\cup K)$}. It is known that the $\nu_0(P\cup K) \ge\lceil\frac{n+m-1}{3}\rceil$; see~\cite{Babu2013}. By Inequality~(\ref{ineq5}), $$n\le (n+m)-\lceil\frac{n+m-1}{3}\rceil\le \frac{2(n+m)+1}{3},$$ and consequently $m\ge \lceil\frac{n-1}{2}\rceil$ (note that $m$ is an integer number).
\end{proof}

Figure~\ref{blocking-fig}(a) shows a \kTD{0}{} graph on a set of 12 points which is blocked by 6 points. By removing the topmost point we obtain a set with odd number of points which can be blocked by 5 points. Thus, the lower bound provided by Theorem~\ref{blocking-0TD-thr} is tight. 

Now let $k=1$. By Theorem~\ref{matching-1TD} we have $\nu_1(P\cup K)\ge \frac{2((n+m)-1)}{5}$, and by Inequality~(\ref{ineq5}) $$n\le (n+m)-\frac{2((n+m)-1)}{5}=\frac{3(n+m)+2}{5},$$ and consequently $m\ge \lceil\frac{2(n-1)}{3}\rceil$; the same lower bound as in Theorem~\ref{blocking-thr1}. 

Now let $k=2$. By Theorem~\ref{mt-thr} we have $\nu_2(P\cup K)= \lfloor\frac{n+m}{2}\rfloor$ (note that $n+m$ may be odd). By Inequality~(\ref{ineq5}) $$n\le (n+m)-\lfloor\frac{n+m}{2}\rfloor=\lceil\frac{n+m}{2}\rceil,$$ and consequently $m\ge n$, where $n+m$ is even, and $m\ge n-1$, where $n+m$ is odd.  

\begin{theorem}
\label{blocking-thr2}
 There exists a set $K$ of $n-1$ points that blocks \kTD{0}{$(P)$}.
\end{theorem}
\begin{proof}
Let $d^0(p,q)$ be the Euclidean distance between $l^0_p$ and $l^0_q$. Let $\delta = \min\{d^0(p,q): p,q\in P\}$.
 For each point $p\in P$ let $p(x)$ and $p(y)$ respectively denote the $x$ and $y$ coordinates of $p$ in the plane. Let $p_1, \dots, p_n$ be the points of $P$ in the increasing order of their $y$-coordinate. Let $K=\{p'_i: p'_i(x)=p_i(x), p'_i(y)=p_i(y)+\epsilon, \epsilon<\delta, 1\le i\le n-1\}$. See Figure~\ref{blocking-fig}(b). For each point $p_i$, let $E_{p_i}$ (resp. $\overline{E_{p_i}}$) denote the edges of \kTD{0}{$(P)$} between $p_i$ and the points above $l^0_{p_i}$ (resp. below $l^0_{p_i}$). It is easy to see that the downward triangle between $p_i$ and any point $q$ above $l^0_{p_i}$ (i.e. any point $q\in C^1_{p_i}\cup C^2_{p_i}\cup C^3_{p_i}$) contains $p'_i$. Thus, $p'_i$ blocks all the edges in $E_{p_i}$. In addition, the edges in $\overline{E_{p_i}}$ are blocked by $p'_1, \dots, p'_{i-1}$. Therefore, all the edges of \kTD{0}{$(P)$} are blocked by the $n-1$ points in $K$.
\end{proof}
Note that the bound of Theorem~\ref{blocking-thr2} is tight, because \kTD{0}{$(P)$} can be a path representing $n-1$ disjoint triangles and for each triangle we need at least one point to block its corresponding edge. 
We can extend the result of Theorem~\ref{blocking-thr2} to \kTD{k}{$(P)$} where $k\ge 1$. For each point $p_i$ we put $k+1$ copies of $p'_i$ very close to $p_i$. Thus, 

\begin{corollary}
 There exists a set $K$ of $(k+1)(n-1)$ points that blocks \kTD{k}{$(P)$}.
\end{corollary}
\section{Conclusion}
\label{conclusion}
In this paper, we considered some combinatorial properties of higher-order triangular-distance Delaunay graphs of a point set $P$. We proved that
\begin{itemize}

  \item \kTD{k}{} is $(k+1)$ connected.
  \item \kTD{2}{} contains a bottleneck biconnected spanning graph of $P$.
  \item \kTD{7}{} contains a bottleneck Hamiltonian cycle and \kTD{5}{} may not have any.
  \item \kTD{6}{} contains a bottleneck perfect matching and \kTD{5}{} may not have any.
  \item \kTD{1}{} has a matching of size at least $\frac{2(n-1)}{5}$.
  \item \kTD{2}{} has a perfect matching when $P$ has an even number of points.
  \item $\lceil\frac{n-1}{2}\rceil$ points are necessary to block \kTD{0}{}.
  \item $\lceil\frac{(k+1)(n-1)}{3}\rceil$ points are necessary and $(k+1)(n-1)$ points are sufficient to block \kTD{k}{}.
\end{itemize}

We leave a number of open problems:
\begin{itemize}

  \item What is a tight lower bound for the size of maximum matching in \kTD{1}{}?
  \item Does \kTD{6}{} contain a bottleneck Hamiltonian cycle?
 \item As shown in Figure~\ref{TD}(a) \kTD{0}{} may not have a Hamiltonian cycle. For which values of $k=1,\dots, 6$, is the graph \kTD{k}{} Hamiltonian?
\end{itemize}

\bibliographystyle{abbrv}
\bibliography{Higher-Order-TDDEL.bib}
\end{document}